\newif\ifstoc
\stoctrue 
\stocfalse

\ifstoc
  \documentclass{sig-alternate}
  \usepackage[obey]{fixacm}  
  \usepackage{times}
\else
  \documentclass[letterpaper,11pt]{article}

  \usepackage{typearea}
  \typearea{15}

  \usepackage[compact]{titlesec}
\fi

\newcommand{\stocoption}[2]{{\ifstoc #1 \else #2 \fi}}

\usepackage{theorem,latexsym,graphicx}
\usepackage{amsmath,amssymb,enumerate}
\usepackage{boxedminipage}
\usepackage{xspace}
\usepackage{bm}
\usepackage{ifpdf}
\usepackage{color}
\usepackage{algorithm}
\usepackage{algorithmic}
\usepackage{subfigure}
\usepackage{verbatim}
\usepackage{mathrsfs}
\usepackage{paralist}
\usepackage{algorithm,algorithmic}

\ifstoc
\else
\fi

\allowdisplaybreaks

\ifstoc
\newcommand{\lref}[2][]{{#1~\ref{#2}}}
\else
\definecolor{Darkblue}{rgb}{0,0,0.4}
\definecolor{Brown}{cmyk}{0,0.81,1.,0.60}
\definecolor{Purple}{cmyk}{0.45,0.86,0,0}
\newcommand{\mydriver}{hypertex}
\ifpdf
 \renewcommand{\mydriver}{pdftex}
\fi
\usepackage[breaklinks,\mydriver]{hyperref}
\hypersetup{colorlinks=true,
            citebordercolor={.6 .6 .6},linkbordercolor={.6 .6 .6},%
citecolor=Darkblue,urlcolor=black,linkcolor=Darkblue,pagecolor=black}
\newcommand{\lref}[2][]{\hyperref[#2]{#1~\ref*{#2}}}
\fi

\newtheorem{theorem}{Theorem}[section]
\newtheorem{definition}[theorem]{Definition}

\newtheorem{lemma}[theorem]{Lemma}

\newtheorem{claim}[theorem]{Claim}

\newtheorem{corollary}[theorem]{Corollary}

\numberwithin{algorithm}{section}

\ifstoc
\else
\newenvironment{proof}{

\noindent{\bf Proof:}}
{\hfill$\blacksquare$

}

\fi
\newenvironment{proofof}[1]{

\noindent{\bf Proof of {#1}:}}
{\hfill$\blacksquare$

}

\newcommand{\junk}[1]{}
\newcommand{\ignore}[1]{}

\newcommand{\Z}[0]{{\ensuremath{\mathbb{Z}}}}

\newcommand{\sse}{\subseteq}

\newcommand{\C}{{\mathscr{C}}}

\newcommand{\Run}{\mathcal{R}}
\newcommand{\Init}{{\tt Init}}

\newcommand{\deficit}{{\tt deficit}}
\newcommand{\val}{{\tt val}}

\newcommand{\eps}{\varepsilon}

\newcommand{\Deficit}{\Phi}
\newcommand{\Round}{{\tt Round}}
\renewcommand{\sp}{{\hspace*{0.1 in}}}
\newcommand{\bb}{{\boldsymbol {\beta}}}
\newcommand{\symdif}{\triangle}

\newcommand{\Opt}{\ensuremath{\mathsf{opt}\xspace}}
\newcommand{\OPT}{\ensuremath{\mathsf{opt}\xspace}}

\newcommand{\rank}{\ensuremath{\mathsf{rank}\xspace}}
\newcommand{\cost}{\ensuremath{\mathsf{cost}\xspace}}
\newcommand{\weight}{\ensuremath{\mathsf{Wt}\xspace}}

\newcommand{\vrank}{\boldsymbol{\nu}}
\newcommand{\len}{\textsf{len}}
\renewcommand{\rank}{\boldsymbol{\rho}}

\newcommand{\const}{6}

\stocoption{}{}

\newcounter{note}[section]

\newcommand{\qedsymb}{\hfill{\rule{2mm}{2mm}}}
\renewenvironment{proof}{\begin{trivlist} \item[\hspace{\labelsep}{\bf
\noindent Proof.\/}] }{\qedsymb\end{trivlist}}%

\newcommand{\initOneLiners}{%
    \setlength{\itemsep}{0pt}
    \setlength{\parsep }{0pt}
    \setlength{\topsep }{0pt}
}
\newenvironment{OneLiners}[1][\ensuremath{\bullet}]
    {\begin{list}
        {#1}
        {\initOneLiners}}
    {\end{list}}

\newcommand{\squishlist}{
 \begin{list}{$\bullet$}
  { \setlength{\itemsep}{0pt}
     \setlength{\parsep}{3pt}
     \setlength{\topsep}{3pt}
     \setlength{\partopsep}{0pt}
     \setlength{\leftmargin}{1.5em}
     \setlength{\labelwidth}{1em}
     \setlength{\labelsep}{0.5em} } }

\newcommand{\squishend}{
  \end{list}  }

\begin{document}

\title{The Power of Deferral: \\ Maintaining a Constant-Competitive Steiner Tree Online}

\author{
Albert Gu\thanks{Department of Mathematical Sciences, Carnegie Mellon
  University, Pittsburgh, PA 15213. Supported by a Knaster-McWilliams
  Scholarship.}
\and
Anupam Gupta\thanks{Computer Science Department, Carnegie Mellon
    University, Pittsburgh, PA 15213, USA. Supported in part by
    NSF awards CCF-0964474 and CCF-1016799, and by the
    CMU-MSR Center for Computational Thinking.}
\and Amit Kumar\thanks{Dept. of Computer Science and Engg., IIT Delhi,
  India 110016.}
}
\date{\today}
\maketitle

\begin{abstract}
  In the online Steiner tree problem, a sequence of points is revealed
  one-by-one: when a point arrives, we only have time to add a single
  edge connecting this point to the previous ones, and we want to
  minimize the total length of edges added. Here, a tight bound has been
  known for two decades: the greedy algorithm maintains a tree whose
  cost is $O(\log n)$ times the Steiner tree cost, and this is best
  possible. But suppose, in addition to the new edge we add, we have
  time to change a single edge from the previous set of edges: can we do
  much better? Can we, e.g., maintain a tree that is constant-competitive?

  We answer this question in the affirmative. We give a primal-dual
  algorithm that makes only a single swap per step (in addition to
  adding the edge connecting the new point to the previous ones), and
  such that the tree's cost is only a constant times the optimal cost.
  Our dual-based analysis is quite different from previous primal-only
  analyses. In particular, we give a correspondence between radii of
  dual balls and lengths of tree edges; since dual balls are associated
  with points and hence do not move around (in contrast to edges), we
  can closely monitor the edge lengths based on the dual radii. Showing
  that these dual radii cannot change too rapidly is the technical heart
  of the paper, and allows us to give a hard bound on the number of
  swaps per arrival, while maintaining a constant-competitive tree at
  all times.  Previous results for this problem gave an algorithm that
  performed an \emph{amortized} constant number of swaps: for each $n$,
  the number of swaps in the first $n$ steps was $O(n)$. We also give a
  simpler tight analysis for this amortized case.
  
\end{abstract}

\section{Introduction}
\label{sec:intro}

In the online Steiner tree problem, a sequence of points $0, 1, 2,$
$\ldots, n, \ldots$ is revealed online. When the point $i$ arrives we are
told the distances $d(i, j)$ for all $j < i$; the distances between
previous points do not change, and we are guaranteed that the distances
always satisfy the triangle inequality.  The goal is to maintain a tree
spanning all the arrivals and having a small cost (which is the sum of
the lengths of the tree edges). As is usual in online algorithms, all
decisions are irrevocable: once an edge is bought it cannot be removed.
This naturally captures a situation where we are building a network, but
only have time to add a single edge at a time.  The greedy algorithm,
upon  arrival of the $i^{th}$ vertex, greedily attaches it to its
closest preceding point; \cite{IW91,AlonAzar} showed that this algorithm
produces a tree that is $O(\log n)$-competitive against the best
spanning tree on $\{ 1, 2, \ldots, n\}$, for every $n$.  They also
showed a matching lower bound of $\Omega(\log n)$ on the competitive
ratio.

But what if the decisions were not irrevocable? What if, when a new
vertex arrived, we were allowed to add a new edge, but also to swap a
small number of previously-added edges for new ones? Given the power of
hindsight, we could do better---but by how much? Imase and Waxman~ \cite{IW91} showed a
natural greedy algorithm that maintains a $2$-competitive tree and makes
at most $O(n^{3/2})$ swaps over the course of the first $n$ arrivals,
for every $n$. Hence the \emph{amortized budget}, the average number of
swaps per arrival, is $O(n^{1/2})$.  This was substantially improved
upon recently, when Megow, Skutella, Verschae, and
Wiese~\cite{MSVW12,Vers12} gave an algorithm with a constant amortized budget
bound. Specifically, given $\eps > 0$, their algorithm
maintains a tree that is $(1+\eps)$-competitive against the minimum
spanning tree, and performs $O(n/\eps \log 1/\eps)$ swaps in $n$ steps.


Note that both these prior results work in the amortized setting: what if
we could only do a constant number of changes per arrival? \emph{In fact, what
if we only had time to perform a single swap per timestep: could
we maintain a spanning tree that is constant competitive against the
best Steiner tree?} The algorithms used in previous papers do not have
this property, as there exist instances where a single arrival can cause
their algorithms to do a linear number of swaps. The main result of this
paper is an affirmative answer to the above question. 
\begin{theorem}[Constant Budget Algorithm]
  \label{thm:main1}
  There is an online algorithm for metric Steiner tree that performs a
  single edge swap upon each arrival, and maintains a spanning tree with
  cost at most a constant times that of the optimal Steiner tree on the
  current set of points.
\end{theorem}
In fact, we can 
maintain a $2^{O(1/\delta)}$-approximate tree and perform only one
swap per $1/\delta$ rounds \stocoption{(details in the final version).}{(see \lref[Theorem]{thm:main1-str}).}

We discuss the ideas behind the algorithm in
\lref[Section]{sec:intuition}; at a high level, the algorithm is
based on the primal-dual method; our analysis is based on relating the
edges in the tree to dual balls, and tracking the changes via changes in
these dual values. This dual-based analysis is substantially different
from the primal-only analyses used in previous works, and we feel it
gives a better insight into the problem.

Our techniques also allow us to give a trade-off between the number of
swaps and the competitiveness: in fact, we first show a weaker result
that performs a constant number of swaps per arrival and maintains a
constant competitive tree (in \lref[Section]{sec:primal-dual}
and~\ref{sec:analysis}). In \stocoption{the full version}{\lref[Section]{sec:one-swap}}, we
show how refinements of our arguments can reduce the number of swaps and prove
\lref[Theorem]{thm:main1} and its extension mentioned
above.

Our second result is a simpler and improved amortized-budget analysis of
the greedy algorithm studied by~\cite{IW91, MSVW12}. For any $\eps \in
(0,1]$, consider the online algorithm $\mathcal{B}_{1+\eps}$ that
greedily connects each new vertex to the closest previous vertex, and
that also swaps a tree edge $e$ for a non-edge $f$ whenever $\len(e) \ge
(1+\eps)\,\len(f)$ and $T + f - e$ is also a spanning tree. By
construction, $\mathcal{B}_{1+\eps}$ maintains a tree that is
$(1+\eps)$-competitive against the best spanning tree.
\begin{theorem}[Tight Amortized Budget Algorithm]
  \label{thm:main2}
  For any $\eps \in (0,1]$, the algorithm $\mathcal{B}_{1+\eps}$ makes
  at most $n \cdot \log_{1+\eps} 4 \leq 2n/\eps$ swaps over the course of
  $n$ arrivals.
\end{theorem}
This result is asymptotically tight, as a lower bound of
$\Omega(n/\eps)$ is known~\cite{Vers12}. The previous best amortized
bound was $O(n/\eps \log 1/\eps)$ given by~\cite{MSVW12}, for a variant
of $\mathcal{B}_{1+\eps}$ which did not perform all possible
$(1+\eps)$-swaps. 
\stocoption{Due to space constraints, we defer the proof of
  \lref[Theorem]{thm:main2} 
  to the full version of the paper.}{The proof of
  \lref[Theorem]{thm:main2} appears in \lref[Section]{sec:all-swaps}. In \lref[Section]{sec:lbd} we give an
instance where algorithm $\mathcal{B}_2$ needs at least $1.25 n$ swaps;
no instances were known earlier where more than $n$ swaps were needed
for $\eps = 1$.}


\subsection{The Constant-Budget Algorithm: Ingredients, Intuition, and Ideas}
\label{sec:intuition}

One of the main difficulties in analyzing ``primal'' algorithms that
directly deal with edge lengths of the current tree is that swaps in the
tree are not local: two close-by edges may be swapped for two edges that
are far from each other, and spatio-temporal relationships between them
become difficult to reason about. Instead we take a primal-dual approach
that talks about duals around vertices---since vertices do not move, we
can argue about them more easily. The rest of this section outlines the
steps and the intuition behind them; the algorithm itself is summarized
in \lref[Section]{sec:nutshell}.

Say vertices $\{0,1,2, \ldots, i-1\}$ had arrived previously, and now vertex $i$ arrives. We first run a \emph{clustering process}
on the vertices in $[i]= \{0,1, \ldots, i\}$ 
(described in \lref[Section]{sec:clustering}): this is
similar to the moat-growing process in the
Agrawal-Klein-Ravi/Goemans-Williamson primal-dual algorithm, but here
we grow the clusters in discrete exponentially-sized steps.  
This clustering process defines integer ``ranks'' $\rank_i(v)$ for
vertices $v \in [i]$. 
Then we run
a \emph{tree-forming process} using these vertex ranks, which
outputs a tree $T^\star_i$ on $[i]$ that is a constant-approximation to
the optimal Steiner tree on $[i]$. This is where ranks are useful:
we ensure a correspondence between lengths of tree edges and vertex
ranks---a vertex of rank $k$ corresponds to some tree edge of length
$\approx \alpha^k$ (for some small constant $\alpha$), and so we can shift
our focus from tracking edge lengths to tracking vertex ranks. Our
clustering ensures that the ranks of existing vertices never increase as
future arrivals occur, and also that the total number of rank decrements
for all vertices over the course of $i$ arrivals is $O(i)$.
Furthermore, the tree-formation satisfies a Lipschitz property: if $s$
rank decrements occur due to the arrival of vertex $i$, then
$T^\star_{i}$ can be obtained from $T^\star_{i-1}$ by adding an edge
connecting vertex $i$ to its closest vertex in $[i-1]$, and then
performing at most $s$ edge swaps.  Putting these two facts together
gives an $O(1)$-\emph{amortized}-budget and $O(1)$-competitive
algorithm.

But we had promised a constant-\emph{worst-case}-budget algorithm; we
cannot directly use the algorithm above, since some arrivals may cause the ranks
of a linear number of vertices to drop. However, we can fix things, and
this is where the advantages of our dual-based approach become apparent.
In addition to the ranks, we also maintain \emph{virtual ranks}
$\vrank_i(v)$ for vertices $v \in [i]$, which also drop monotonically,
and which are upper bounds on the ranks. And we run the tree-forming
process on these virtual ranks to get the actual tree $T_i$. We want the
virtual ranks to be close to the real ranks, but also not to change too
drastically upon arrivals---so when vertex $i$ arrives, we define
$\vrank_i(\cdot)$ in such a way that
\begin{OneLiners}
\item $\vrank_i(v) \leq \vrank_{i-1}(v)$ for all $v \in [i-1]$ (virtual
  ranks are also monotone decreasing),

\item $\vrank_i(v) \geq \rank_i(v)$ for all $v \in [i-1]$ (virtual ranks
  are upper bounds for actual ranks),

\item $\vrank_i(i) = \rank_i(i)$ (the virtual rank equals the actual
   rank for new arrivals), and

\item $\| \vrank_i - \vrank_{i-1} \|_1 \leq O(1)$ (the number of virtual
    rank changes is only a constant).
\end{OneLiners}
By this last property and the Lipschitz-ness of our tree-formation
algorithm, the number of edge swaps to get from $T_{i-1}$ to $T_i$ is a
constant. 

But what about the competitiveness of the tree?
The cost of our tree $T_i$ is $\approx \sum_{v \in [i]}
\alpha^{\vrank_i(v)}$, whereas the ideal tree $T^*_i$ has cost $\approx
\sum_{v \in [i]} \alpha^{\rank_i(v)}$. Since we want the former sum to
be close to the latter, we define the virtual ranks by decrementing it
for those nodes $v$ for which $\vrank_{i-1}(v) > \rank_i(v)$ and the
numerical value of $\vrank_{i-1}(v)$ is the largest. The technical heart
of the paper lies in showing that this way of maintaining virtual ranks
gives us a constant approximation tree at all times; the analysis is
given in \lref[Section]{sec:analysis}.

\subsubsection{The Algorithm in a Nutshell}
\label{sec:nutshell}

Most of the above description is intuition and analysis. Our algorithm
is simply the following. When vertex $i$ arrives,
\begin{OneLiners}
\item[(i)] run clustering to get ranks $\rank_i(v)$ for all $v \in [i]$,
\item[(ii)] define the virtual ranks $\vrank_i(v)$ using the simple
  greedy rule
  described above, and
\item[(iii)] run tree-formation on the virtual rank function $\vrank_i$
  to get the tree $T_i$.
\end{OneLiners}
We emphasize that the clustering and tree-formation algorithms are just
two halves of the Agrawal-Klein-Ravi/Goe\-mans-Williamson
moat-growing primal-dual algorithm, by viewing the moat-growing and
edge-additions separately.

\subsection{Related Work}
\label{sec:related-work}

The online Steiner tree problem was studied by Imase and
Waxman~\cite{IW91}, who proved $\Theta(\log n)$-competi\-tive\-ness for the
problem when no swaps are allowed. The proof was simplified by Alon and
Azar~\cite{AlonAzar}, who also gave an $\Omega(\frac{\log n}{\log\log
  n})$-lower bound for planar point sets. Imase and Waxman also gave an
algorithm that maintained a $2$-competitive tree and used an amortized
budget of $O(\sqrt{n})$ over $n$ steps. They also considered the model
where vertices could leave the system: for this problem they also gave
an algorithm with the same amortized budget, but a weaker
$8$-competitiveness.

The primal-dual method has been a powerful tool in offline algorithm
design; see, e.g., the treatment in the textbooks~\cite{vvv,SW10}. This
technique was used in original Steiner forest papers of Agrawal, Klein,
and Ravi~\cite{AKR95} and Goemans and Williamson~\cite{GW95}. The
popularity and power of primal-dual in online algorithm design is more
recent (see, e.g., the monograph of Buchbinder and Naor~\cite{BN-mono}
for many successes).  Some early uses of primal-dual ideas in online
algorithms can be seen in the online Steiner forest analysis of
Awerbuch, Azar, and Bartal~\cite{AAB96}, and the algorithm and analysis
of Berman and Coulston~\cite{BC97}.

Apart from the results of~\cite{IW91, MSVW12}, the idea of making a
small number of alterations to maintain a good solution in online
settings had been studied for other problems. See, e.g., the works
of~\cite{Wes00, AGZ99, AAPW01, SSS, SV10, EL11} which study alterations in the
context of online scheduling and routing problems.
Ashwinkumar~\cite{Ashwin11} gives optimal results for the problem of
maximizing the value of an independent set in the intersection of $p$
matroids, where the elements of the universe arrive online; in his
model, an element can get added to the set and subsequently canceled (for a
price), but an element that has been dropped can never be added
subsequently. However, the ideas in these papers seem to be orthogonal
to ours.

Our results are also related, at a conceptual level, to those of Kirsch
and Mitzenmacher~\cite{KM07}, Arbitman et al.~\cite{ANS09} and related
works on \emph{de-amortizing} data structures (such as cuckoo
hash-tables) by maintaining a suitable auxiliary data structure (e.g., a
queue or a ``stash''); this allows them to achieve constant update times
with overwhelming probability. Since our goal is slightly different---to
achieve low cost---our de-amortization works by delaying certain updates
via a priority queue.


\section{A Constant-Swaps Algorithm}
\label{sec:primal-dual}

It will be useful to first prove a slightly weaker version of
\lref[Theorem]{thm:main1}, which will introduce the main
ideas.
\begin{theorem}
  \label{thm:weaker}
  Let $\alpha \geq 6$. There is an online algorithm that makes at most
  $K = 2\alpha^2$ swaps upon each vertex arrival, and for every $n$, maintains a tree
  $T_n$ with cost at most $C = \smash{\frac{2\alpha^5}{(\alpha-1)^2}}$
  times the cost of an optimal Steiner tree on the first $n$ points. 
\end{theorem}
We describe the algorithm of \lref[Theorem]{thm:weaker} in this section,
and give its analysis in \lref[Section]{sec:analysis}. 
\stocoption{Due to space constraints, the proof of the single-swap
  algorithm is deferred to the full version of the paper.}{We then show how to
modify the algorithm slightly to trade off swaps for performance, and
hence get a single-swap algorithm in \lref[Section]{sec:one-swap}.}

In order to describe the online algorithm,
let us lay down some notation.
Let $[n]$  denote the integers $\{0, 1, 2, \ldots,
n\}$, and $[i \ldots j]$  denote the integers $\{i, i+1, \ldots, j\}$.
We associate the arriving vertices in the metric space with the
integers: we start off with the root vertex $0$, and the $n^{th}$
arriving vertex is called $n$; and hence the root and the first $n$
arriving vertices are identified with the set $[n]$.

For $i \geq 1$, when the $i^{th}$ vertex arrives, we begin \emph{round}
$i$. Round $i$ involves three steps:
\begin{OneLiners}
\item[(a)] running a clustering algorithm on the vertex
  set $[i]$ (itself involving several ``phases'') which defines
  the rank function $\rank_i: [i] \to \Z_{\geq 0}$ (described in
  \lref[Section]{sec:clustering}),
\item[(b)] getting a virtual rank function $\vrank_i$ from the actual
  rank function $\rank_i$ (as in \lref[Section]{sec:def-vrank}), and
\item[(c)] finally constructing the tree $T_i$ given the virtual rank function
  (described in \lref[Section]{sec:timed-pd}).
\end{OneLiners}
This concludes round $i$. We denote $\Run_i$ to be the run of the
clustering algorithm for round $i$; as mentioned above it has several
phases. Let $\Opt([i])$ denote the cost of the optimal Steiner tree on
$[i]$. We want to relate the cost of $T_i$ to $\Opt([i])$.

For a vertex $x \in [n]$ and subset $S \subseteq [n]$, define $d(x,S) :=
\min_{s \in S} d(x,s)$.  Similarly, let $d(S,T) = \min_{s \in S}
d(s,T)$.  For \emph{radius} $r \geq 0$, define the \emph{ball} $B(x,r) =
\{ y \in [n] \mid d(x,y) \leq r\}$.  For a set $S \sse [n]$ and radius
$r \geq 0$, let $B(S,r) = \cup_{s \in S} B(s,r) = \{y \mid d(y,S) \leq r
\}$.


\subsection{The Clustering Procedure}
\label{sec:clustering}

Let $\alpha$ be a universal constant $\geq \const$. We assume that all
inter-vertex distances are at least $2\alpha$, else we can scale things
up. The run $\Run_i$ on vertex set $[i]$ starts off with trivial
clustering $\C_i(0)$ containing $i+1$ disjoint \emph{clusters} $\{\{0\},
\{1\}, \ldots, \{i\}\}$. At the beginning of phase $t$ of $\Run_i$, we
have a clustering $\C_i(t-1)$, which is a partition of $[i]$, produced
by phase $t-1$. The invariant is that two distinct clusters $C, C' \in
\C_i(t-1)$ have $d(C,C') \geq 2\alpha^t$.  (Since all distances are at
least $2\alpha$, the clustering $\C_i(0)$ satisfies this invariant for
the first phase $t=1$.) We start with the clusters in $\C_i(t-1)$, and
while there exist two clusters $C$ and $C'$ that satisfy $d(C,C') <
2\alpha^{t+1}$, we \emph{merge} these two clusters into one (i.e., we
remove $C$ and $C'$ and add in $C \cup C'$). Note the resulting
clustering, which we call $\C_i(t)$, satisfies minimum inter-cluster distance
$2\alpha^{t+1}$ by construction. This defines clusterings $\C_i(t)$ for
all $t \in \Z_{\geq 0}$.

For each cluster $C \in \C_i(t)$, define $C$'s \emph{leader} as the vertex
in $C$ with the least index.  Note that for $t = 0$ each vertex is a leader,
and as $t$ gets very large, only the vertex~$0$ remains the leader.  Let
the \emph{rank} $\rank_i(x)$ of vertex $x \in [1\ldots i]$ be the
largest $t$ for which $x$ is the leader of its cluster in $\C_i(t)$ (i.e., at the end of
phase $t$); define the rank
of vertex $0$ as $\rank_i(0):= \infty$. Finally, for a value $j \geq 0$
and a function $f: [j] \to \Z_{\geq 0}$, define its weight
\begin{gather}
  \weight_j(f) := \sum_{l = 1}^{j} \alpha^{f(l)}. \label{eq:wt}
\end{gather}
(Note the definition of $\weight_j(\cdot)$ \emph{does not} include
the index~$0$ in the sum.)

\subsubsection{Properties of the Ranks and Clusterings}
\label{sec:rank-ppties}

We now prove some useful properties about ranks and clusterings: these
are not needed for the algorithm, only for the analysis.

\begin{lemma}
  \label{lem:lbd}
  For any $i \geq 1$, we have $\weight_i(\rank_i) \leq
  \frac{\OPT([i])}{\alpha - 1}$.
\end{lemma}

\begin{proof}
  Recall the dual LP for the natural relaxation of the Steiner tree
  problem. We consider the graph $G = ([i], \binom{[i]}{2})$, with the
  length of the edge $(i,j)$ being $d(i,j)$. The dual says:
  \begin{align}
    \max \sum_S y_S & \tag{$DLP_{ST}$} \\
    \sum_{S : |S \cap \{j,l\}| = 1} y_S &\leq  d(j,l) \qquad\qquad\forall
    j, l \in [i] \label{eq:dualLP} \\
    y_S &\geq 0. \notag
  \end{align}
  We shall define a feasible solution $y$ such that $\sum_S y_S \geq
  \weight_i(\rank_i)\cdot(\alpha-1)$; by weak duality this will imply that $\OPT([i]) \geq
  \weight_i(\rank_i)\cdot(\alpha-1)$. Consider the clustering defined by  the run $\Run_i$. For
  every $t \geq 0$, and every cluster $C \in \C_i(t)$ such that vertex
  $0 \notin C$, we define $y_C = \alpha^{t}(\alpha -1)$. For all other sets $S$,
  set $y_S = 0$.

  To check feasibility, consider any edge $(j,l)$, and let $t$ be the
  last phase such that $j$ and $l$ lie in different clusters in
  $\C_i(t)$. For all phases $t' \leq t$, we contribute $2\alpha^{t'}(\alpha-1)$
  towards the left hand side of (\ref{eq:dualLP}). Moreover, $d(j,l)
  \geq 2 \alpha^{t+1}$, because they are in different clusters in
  $\C_i(t)$). Hence, the LHS of~(\ref{eq:dualLP}) is $\sum_{t'=0}^{t} 2
  \alpha^{t'}(\alpha-1) = 2(\alpha^{t+1}-1) \leq d(j,l)$.

  Now consider the objective function value $\sum_C y_C$, which we claim
  is at least $\weight(\rank_i) \cdot (\alpha-1)$. Indeed, consider the
  following map $g$ from $[1 \ldots i]$ to $\cup_t \C_i(t)$: for vertex $j
  \in [1 \ldots i]$, let $g(j)$ be the cluster in $\C_i(\rank_i(j))$
  for which $j$ is the leader. Since this is a 1-1 mapping,
  $$ \sum_C y_C \geq \sum_{j=1}^i y_{g(j)} = \sum_{j=1}^i
  \alpha^{\rank_i(j)}(\alpha-1) = \weight_i(\rank_i)\cdot(\alpha-1). $$ This proves the lemma.
\end{proof}

The following lemma shows that if a set of vertices $S$ is far from the rest of the vertices, then
until a high enough phase $t$, any cluster in $\C_i(t)$ will be a subset
of $S$ or of $[i] \setminus S$.
\begin{lemma}
  \label{lem:empty}
  Suppose $S \sse [i]$ such that $B(S, 2 \alpha^t) \cap [i] = S$ for
  some value $t$. Then for any phase $k \leq t-1$ and any cluster $C$ in
   $\C_i(k)$, either $C \subseteq S$ or $C \cap
  S = \emptyset$.
\end{lemma}

\begin{proof}
  Proof by induction on the phases of run $\Run_i$. For the base case, each
  cluster in $\C_i(0)$ is a singleton and the claim holds. Now suppose
  for some phase $k < t-1$, each $C \in \C_i(k)$ either lies within $S$
  or is disjoint from it. Note that the assumption  $B(S, 2\alpha^t)
  \cap [i] = S$ is  same as saying $d(S, [i]\setminus S) >
  2\alpha^t$. So, if $C, C' \in \C_i(k)$ satisfy $C \subseteq S$ and $C'
  \cap S = \emptyset$, then $d(C,C') > 2 \alpha^t \geq 2\alpha^{k+1}$,
  and we will not merge these two clusters in $\C_i(k+1)$. Hence, the
  claim holds for phase $k+1$ as well.
\end{proof}

The clusterings produced by the runs $\Run_{i-1}$ on the set $[i-1]$,
and $\Run_{i}$ on $[i]$ are closely related:
the clustering $\C_{i-1}(t)$ is a refinement of the clustering
$\C_{i}(t)$, as the next lemma shows.

\begin{lemma}
  \label{lem:relating}
  For a cluster $C \in
  \C_{i}(t)$, exactly one of the following holds:
  \begin{OneLiners}
  \item[(a)] vertex $i \not\in C$: in this case
    $C$ is also a cluster in $\C_{i-1}(t)$, or
  \item[(b)] vertex $i \in C$: in this case $C = \{i\} \cup C_1
    \cup C_2 \cup \ldots \cup C_p$ for some $p \geq 0$ clusters $C_1, \ldots,
    C_p \in \C_{i-1}(t)$.
  \end{OneLiners}
  (In the latter case, note that $p$ may equal $0$, in which case $C = \{i\}$.)
\end{lemma}

\begin{proof}
  The proof is again by induction on $t$. At $t=0$, this is  true
  because all clusters are singleton elements.  Suppose the claim of
  the lemma holds for some phase $t \geq 0$. Let $C_1, \ldots, C_l$ be 
  the clusters in $\C_{i-1}(t)$. By the induction hypothesis  
  we can renumber these clusters in $\C_{i-1}(t)$ such  that 
  the clusters in $\C_i(t)$ are $\{i\} \cup C_1 \cup \ldots \cup C_p, 
  C_{p+1}, \ldots, C_l$ for some $p \geq 0$. We construct an auxiliary graph $G_{i-1}(t)$ with 
  vertices corresponding to the clusters in $\C_{i-1}(t)$, and join
  two clusters $C_i, C_j$ by an edge in $G_{i-1}(t)$ if $d(C_i, C_j) \leq 2 \alpha^{t+2}$. 
  By the definition of the clustering process, the clustering $\C_{i-1}(t+1)$ is obtained by taking the 
  unions of the clusters in each connected component of $G_{i-1}(t)$. We can define
  $G_i(t)$ similarly, and again, the clusters in $\C_{i}(t)$ correspond to connected
  components of $G_i(t)$. 

  Now observe that if $i, j > p$, we have an edge between $(C_i, C_j)$
  in $G_i(t)$ exactly when
  we have this edge in $G_{i-1}(t)$. And if $i \leq p, j > p$, and we have the edge 
  $(C_i, C_j) \in G_{i-1}(t)$, then we have an edge between $C_j$ and 
  $\{i\} \cup C_1 \cup C_2 \cup \ldots \cup C_p$ in $G_i(t)$. The hypothesis for $t+1$ follows from these facts. 
\end{proof}



\begin{lemma}
  \label{lem:delayed-merge}
  If $j,l \in [i-1]$ lie in a common cluster in $\C_i(t)$, then in run
  $\Run_{i-1}$ they lie in a common cluster either in $\C_{i-1}(t)$ or in
  $\C_{i-1}(t+1)$.
\end{lemma}

\begin{proof}
  Suppose $j,l \in [i-1]$ lie in a common cluster in $\C_i(t)$, but
  belong to different clusters of $\C_{i-1}(t)$. By
  \lref[Lemma]{lem:relating}, there are a set of clusters $C_1, \ldots,
  C_p \in \C_{i-1}(t)$ such that $C = \{i\} \cup C_1 \cup \ldots \cup
  C_p$ is a cluster in $\C_i(t)$, and these two vertices $j,l \in C$.
  Consider a cluster $C_q$ for $1 \leq q \leq p$.  We claim that $d(i,
  C_q) \leq 2 \alpha^{t+1}$. Suppose not. The fact that $C_q$ is a
  cluster in $\C_{i-1}(t)$ implies that $d(C_q, [i-1] \setminus C_q) \geq 2
  \alpha^{t+1}. $ Then, $d(C_q, [i] \setminus C_q) \geq 2 \alpha^{t+1}$ as
  well. But then, by \lref[Lemma]{lem:empty}, there cannot be a cluster
  in $\C_i(t)$ containing a vertex from both $C_q$ and $[i] \setminus C_q$,
  which contradicts the assumption about the cluster $C \in \C_i(t)$.
  Now applying the triangle inequality, for any $q, q' \in [1 \ldots
  p]$, the distance $d(C_q, C_{q'}) \leq d(i, C_q) + d(i, C_{q'}) \leq 4
  \alpha^{t+1} \leq 2 \alpha^{t+2}$. (Recall that $\alpha \geq 2$.)
  Consequently, we will merge $C_q$ and $C_{q'}$ in phase $t+1$ of the
  run $\Run_{i-1}$. Indeed, all the vertices in $C_1 \cup \ldots \cup
  C_q$ will lie in a common cluster in $\C_{i-1}(t+1)$, which proves the
  lemma.
\end{proof}




\begin{corollary}[Ranks are Monotone]
  \label{cor:rank-fall}
  For $j \in [i-1]$, $\rank_i(j) \leq \rank_{i-1}(j) \leq \rank_i(j) +
  1$. Hence the rank of a vertex is non-increasing as a function of $i$.
\end{corollary}

\begin{proof}
  If $j$ is no longer the leader of its cluster in $\C_{i-1}(t)$
  (because some $l < j$ lies in its cluster), then since $\C_{i-1}(t)$
  is a refinement of $\C_i(t)$, $l$ lies in $j$'s cluster in $\C_i(t)$
  too. Hence $\rank_i(j) \leq \rank_{i-1}(j)$.

  Now suppose $j$ loses leadership of its cluster during phase $t$ of
  $\Run_i$, i.e.,  $\rank_{i}(j) = t-1$. So there exists a vertex   $l < j$ which lies in
  the cluster of $\C_i(t)$ containing $j$.
  \lref[Lemma]{lem:delayed-merge} says $j, l$ must share a
  cluster in $\C_{i-1}(t+1)$, making $\rank_{i-1}(j) \leq t+1$.
\end{proof}

\begin{claim}[Initial Ranks]
  \label{clm:init-rank}
  If the initial rank of vertex $i$ (i.e., $i$'s rank in run $\Run_i$)
  is $r$, then the distance $d(i, [i-1]) \in [2\alpha^{r+1},
  2\alpha^{r+2})$.
\end{claim}

\begin{proof}
  Let $j$ be the closest vertex in $[i-1]$ to $i$. Then if $d(i,j) <
  2\alpha^{r+1}$, then $i,j$ would be part of the same cluster in
  $\C_i(r)$ and hence $\rank_i(i) < r$. Similarly, if $d(i,j) \geq
  2\alpha^{r+2}$, then \lref[Lemma]{lem:empty} shows that  $i$ would form a singleton cluster in $\C_i(r+1)$
  and hence $\rank_i(i) \geq r+1$.
\end{proof}

\subsection{The Tree-Building Process}
\label{sec:timed-pd}

In this section, we explain the second ingredient of our
algorithm: given the rank function $\rank_i,$ how to build a tree $T_i
= ([n], E_i)$. In fact, we give a more general process that takes a
function from a wider class of ``admissible'' functions and produces a
tree for such a function. We want the trees $T_i$ and $T_{i-1}$ to look
similar, so our tree-building procedure assumes access to $T_{i-1}$ when
building tree $T_i$.


Recall that for a vertex $j \leq i$, $\rank_i(j)$ denotes its rank in
the primal-dual process $\Run_i$. Moreover, define $\Init(j) :=
\rank_j(j)$ to be the initial rank of $j$ (when it arrived in round
$j)$; define $\Init(0) = \infty$. We say that a function $\bb: [i] \to
\Z_{\geq 0}$ is {\em admissible} if $\bb(j) \in [\rank_i(j)\ldots \Init(j)]$
for all $j \in [i]$. (Thus the rank function $\rank_i$ is always
admissible.) For a set $S \subseteq [i]$, the \emph{head of $S$} with
respect to the function $\bb$ is defined to be the vertex $j \in S$ with
highest $\bb(j)$ value (in case of ties, say, choose the lowest
numbered vertex among these to be the head).~\footnote{Since $\bb(0)=\infty$, the root 
  vertex will always be the head of the component containing it.}

A tree $T = ([i], E_T)$ is defined to be {\em valid with respect to
  $\bb$} if we can partition the edge set $E_T$ into sets $E_T^1, E_T^2,
\ldots, E_T^r$ such that the following two conditions are satisfied for
each $l \in [1\ldots r]$:
\begin{OneLiners}
\item[(i)] Let $E_T^{\leq l}$ denote $E_T^1 \cup \cdots \cup E_T^l$. For
  any connected component of $E_T^{\leq l}$, let $j$ be the head of this
  component. Then we require $\bb(j) \geq l$.
\item[(ii)] Each edge in $E_T^l$ has length at most $2 \alpha^{l+1}$.
\end{OneLiners}

\begin{lemma}
  \label{lem:cost}
  Let $T$ be any tree valid with respect to $\bb$. Then the total cost of
  $T$ is at most
     $2 \frac{\alpha^3}{\alpha-1} \cdot \weight_i(\bb)$.
\end{lemma}
\begin{proof}
  For any $l \in [1 \ldots r]$, there must be at least $|E_T^l|+1$ connected components in $E_T^{\leq (l-1)}$.
  The cost of each edge in $E_T^l$ can be charged to the heads of the
  components of $E_T^{\leq (l-1)}$, except for the root vertex $0$. Each of these
  head vertices have $\bb(j)$ values
  at least $l-1$ by condition~(i) of validity. Now any vertex $j \neq 0$
  is charged by some $E_T^l$ only if $l \leq \bb(j) + 1$, 
and since each edge in $E_T^l$ has length at most $2
  \alpha^{l+1}$ (by condition~(ii) of validity), the total charge to $j$ is at most
  $$ \sum_{l = 1}^{\bb(j) + 1} 2 \alpha^{l+1} \leq 2 \frac{\alpha^3}{\alpha-1} \cdot
  \alpha^{\bb(j)}. $$
Summing over $j \neq
  0$ and using the definition of $\weight_i(\cdot)$ from~(\ref{eq:wt}) completes the proof.
\end{proof}

We now prove a Lipschitz property of the $\bb$ function: namely, if we
decrement some coordinates of an admissible function $\bb$ to get
another admissible function $\bb'$ at $L_1$ distance $\|\bb -
\bb'\|_1$, then we can change a tree $T$ valid for $\bb$ into a tree
$T'$ valid for $\bb'$ by making at most $\|\bb - \bb'\|_1$ 
swaps. 
\begin{lemma}
  \label{lem:modify}
  Let $\bb$ be an admissible function and $T = ([i], E_T)$ be a valid
  tree with partition $(E_T^1, \ldots, E_T^r)$. Let $j^\star \in [i]$,
  and suppose $\bb'$ satisfies $\bb'(j) = \bb(j)$ if $j \neq j^\star$,
  and $\bb'({j^\star}) = \bb({j^\star})-1$. Assume that $\bb'$ is
  also admissible (i.e., $\rank_i(j^\star) \leq \bb'({j^\star})$). Then
  there is a valid tree $T'=([i], E_{T'})$ with respect to $\bb'$ such
  that $|E_{T'} \symdif E_T| \leq 2. $
\end{lemma}

\begin{proof}
  For brevity, let $l^\star := \bb({j^\star})=\bb'({j^\star}) + 1$. 
  Let us define the tree $T'$ as
  follows. For values $l \leq l^\star - 1$, define $E_{T'}^l := E_T^l$.
  Condition (ii) remains satisfied for these values of $l$ since the
  edge sets are unchanged; moreover, since $l \leq l^\star - 1
  \Rightarrow l \leq \bb'({j^\star})$, even if $j^\star$
  happened to be the head of a component of $E_{T'}^{\leq l}$,
  condition~(i) would be satisfied.

  Next, we initialize set $E_{T'}^{l^\star}$ to contain all the edges in
  $E_T^{l^\star}$. It may however happen that $j^\star$ was the head
  of a connected component $C$ in $E_{T}^{\leq l^\star}$; since
  $\bb'({j^\star}) < l^\star$, this component $C$ would now violate
  condition~(i). In this case, we claim that there must be some vertex
  $j \notin C$ such that $d(j, C) \leq 2 \alpha^{l^\star + 1}$.  Suppose
  not; then \lref[Lemma]{lem:empty} implies that there is a vertex $v \in
  C$ such that $\rank_i(v) \geq l^\star$, and so $\bb'(v) \geq
  l^\star$. But then $j^\star$ cannot be the head of $C$, a
  contradiction. We now add an edge $e^\star$ between the claimed $j \notin C$ and
  its closest vertex in $C$, with the cost of this edge at most $2
  \alpha^{l^\star + 1}$---this completes the description of
  $E_{T'}^{l^\star}$. Note that this satisfies condition~(i), since the
  vertex $j \notin C$ belonged to a component whose head had $\bb$
  value at least $l^\star$ (by the validity of $T$ with respect to
  $\bb$), and adding the edge between $C$ and that
  component fixes the problem for $C$. 

  For $l > l^\star$, we define the edge sets $E_{T'}^l$ as follows: 
   we add all edges of $E_T^{l+1}$ to $E_{T'}^{l+1}$, except for edges that
  connect two vertices in the same component of $E_{T'}^{\leq l}$. Since
  $E_{T'}^{\leq l^\star}$ has one more edge that $E_{T'}^{\leq l^\star}$, there will be 
  at most one edge from $\cup_{l > l^\star} E_{T}^l$ which does not get added to $T'$. 
  Hence, the symmetric
  difference between $T$ and $T'$ has size at most $2$. 

  Now we show that $T'$ is valid with respect to $\bb'$. Condition (ii) is
  easy to check: the only edge in $E_{T'} \setminus E_T$, namely $e^\star$,
   is added in level $l^\star$ and  has 
  length at most $2 \alpha^{l^\star + 1}$. Moreover, all edges in $E_{T'}$ except perhaps for $e^\star$
  have the same levels as the corresponding edges in $E_T$. 

  It remains to check condition (i) for levels at least $l^\star$. To begin, 
  observe the following invariant for all levels $l \geq l^\star$:~each 
  component of
  $E_{T'}^{\leq l}$ consists of union of some of the components of
  $E_{T}^{\leq l}$. Indeed, this holds at level $l^\star$, and subsequently we 
   add all edges of $T$ to $T'$ except for one edge which forms a cycle. 
   Now at level $l^\star$, condition (i) holds by the construction of
   $E_{T'}^{l^\star}$. Now for a level $l > l^\star$, let $C'$ be any
   component of $E_{T'}^{\leq l}$. The invariant above implies that $C'$
   contains a connected component $C$ of $E_T^{\leq l}$. Since $T$ is
   valid with respect to $\bb$, the leader $i$ of $C$ satisfies $\bb(i)
   \geq l$, and so, $i$ is different from $j^\star$. But then $\bb'(i) =
   \bb(i) \geq l$ and hence condition (i) holds for $C'$ as well. This
   completes the proof of the theorem.
\end{proof}

We now show that a Lipschitz property also holds when adding
 a new vertex.
\begin{lemma}
  \label{lem:addition}
  Suppose $T$ is a valid tree on $[i]$ with respect to $\bb$. Consider a
  new function $\bb':[i+1]\to \Z_{\geq 0}$ defined thus: $\bb'(j) :=
  \bb(j)$ if $j \leq i$, and $\bb'(i+1) := \Init(i+1)$. Then, there is a
  valid tree $T'$ with respect to $\bb'$ such that $|T' \symdif T| = 1$.
  Moreover, if $\bb$ was admissible, then so is $\bb'$.
\end{lemma}

\begin{proof}
  Let $l^\star$ denote $\Init(i+1)$. We know that if $j^\star \in [i]$
  is the closest vertex to the new vertex $i+1$, then $2 \alpha^{l^\star
    + 1} < d(j^\star, i) \leq 2 \alpha^{l^\star +2}$ (\lref[Claim]{clm:init-rank}). 
    For $l \neq
  l^\star+1$, we set $E_{T'}^l := E_{T}^l$, and we define
  $E_{T'}^{l^\star+1} := E_{T}^{l^\star+1} \cup \{ (j^\star, i+1) \}$.
  It is easy to verify that $T'$ is valid with respect to $\bb'$.

  Note that if $\bb$ was admissible, then using the facts that the ranks
  of the vertices can never increase, and that $\bb'(i+1) =
  \rank_{i+1}(i+1)$, we get that $\bb'$ is admissible.
\end{proof}

Observe that rank functions $\{\rank_i\}_i$ produced by the clustering
procedure upon each arrival are always admissible. Hence, starting from
$T^\star_{i-1}$, we can add a single edge (from $i$ to its closest
vertex in $[i-1]$) using \lref[Lemma]{lem:addition}, and then perform at
most $\| \rank_i - \rank_{i-1}\|_1$ edge swaps (using
\lref[Lemma]{lem:modify}) to get the tree $T^\star_i$ valid with respect
to $\rank_i$. This tree is constant competitive, because of
\lref[Lemma]{lem:cost} and \lref[Lemma]{lem:lbd}); moreover, the results
in \lref[Section]{sec:analysis} will show that $\sum_{i = 1}^n \|
\rank_i - \rank_{i-1}\|_1 \leq O(n)$, which gives us another
constant-amortized-swaps algorithm.  However, there may be rounds that
perform a non-constant number of swaps, so this does not give us our
final result. To get both constant-worst-case-swaps and constant
competitiveness, we need another admissible function---the
\emph{virtual rank function}---which we define next.






\subsection{Defining the Virtual Rank Function}
\label{sec:def-vrank}


We now describe how to maintain an (admissible) virtual rank function
$\vrank_i$ for all rounds $i$. We will ensure that the $L_1$ distance between 
$\vrank_i$ and
$\vrank_{i-1}$ is at most a  constant number $K$ ---we
can then use \lref[Lemmas]{lem:modify} and~\ref{lem:addition} to
construct a corresponding valid tree $T_{i}$ which differs from
$T_{i-1}$ only in a constant number of edges.  Furthermore, we need to
keep the cost of $T_{i}$, which is $\approx \sum_{j > 0}
\alpha^{\vrank_i(j)}$, as small as possible. A natural way to
obtain $\vrank_{i}$ from $\vrank_{i-1}$ is to iteratively decrease the virtual
rank of those $K$ vertices (which could be a multiset) 
for which $\vrank_{i-1}(j)$ values are highest
(provided $\vrank_{i-1}(j)$ is strictly larger than $\rank_{i}(j)$).

Motivated by this, we define a total ordering on pairs $(j,k)$, where
$v$ is a vertex and $k$ is an integer: we say that $(j,k) \prec (j',k')$
if either $k < k'$, or else $k=k'$ and $j < j'$. We formally give the
algorithm for maintaining virtual ranks in the figure below. 

  \stocoption{\begin{figure}[ht]
  \centering
    \begin{boxedminipage}{\columnwidth}
      {\bf Virtual Ranks  :} \medskip\\
        \sp 1. Initially, we just have the root vertex 0. Define $\vrank_0(0)=\infty. $\\
        \sp 2. For $i= 1, 2, \ldots$ \\
        \sp \sp (i) Run the clustering algorithm $\Run_i$ to define $\rank_i$. \\
        \sp \sp (ii) Set  $\vrank_{i}(i)$ as $\Init(i)$. \\
         \sp \sp (iii) Define 
\[ Q(i) = \{ (j,k) \mid j \in [i-1], k \in [\rank_{i}(j) \ldots (\vrank_{i-1}(j)-1) ] \}. \] \\
         \sp \sp (iv)  Let $Q_K$ be the set of $K$ highest pairs
         (w.r.t. $\prec$) from $Q$. \\
         \sp \sp (v) Define the first $i-1$ coordinates of $\vrank_{i}$ as follows: \\
          \sp \sp \sp \sp $\vrank_{i}(j) := \left\{ \begin{array}{cc}\vrank_{i-1}(j) &
              \mbox{if }  (j, \star ) \notin Q_K \\  \min\{ k \mid (j,
              k) \in Q_K \} &
              \mbox{if } (j, \star) \in Q_K \end{array} \right.$ \\
      \end{boxedminipage}
      \label{fig:recourse}
      \caption{Algorithm maintaining virtual ranks; $K = 2 \alpha^2$.}
  \end{figure}}
  {\begin{figure}[ht]
  \centering
    \begin{boxedminipage}{\textwidth}
      {\bf Virtual Ranks  :} \medskip\\
        \sp \sp 1. Initially, we just have the root vertex 0. Define $\vrank_0(0)=\infty. $\\
        \sp \sp 2. For $i= 1, 2, \ldots$ \\
        \sp \sp \sp \sp (i) Run the clustering algorithm $\Run_i$ to define the rank function $\rank_i$. \\
        \sp \sp \sp \sp (ii) Set  $\vrank_{i}(i)$ as $\Init(i)$. \\
         \sp \sp \sp \sp (iii) Define $Q(i) = \{ (j,k) \mid j \in [i-1], k \in [\rank_{i}(j) \ldots (\vrank_{i-1}(j)-1) ] \}. $ \\
         \sp \sp \sp \sp (iv)  Let $Q_K$ be the set of the $K$ highest pairs
         (w.r.t. $\prec$) from $Q(i)$. \\
         \sp \sp \sp \sp (v) Define the first $i-1$ coordinates of $\vrank_{i}$ as follows: \\
          \sp \sp \sp \sp \sp \sp \sp $\vrank_{i}(j) := \left\{ \begin{array}{cc}\vrank_{i-1}(j) &
              \mbox{if }  (j, \star ) \notin Q_K \\  \min\{ k \mid (j,
              k) \in Q_K \} &
              \mbox{if } (j, \star) \in Q_K \end{array} \right.$ \\
      \end{boxedminipage}
      \label{fig:recourse}
      \caption{Algorithm maintaining virtual ranks; $K = 2 \alpha^2$.}
    \end{figure}}

    An important observation about the definition of $\vrank_i$: the set
    $Q_K$ might contain both tuples $(j,k+1), (j, k)$ for some $j \in
    [i]$ and $k \geq 0$. But if $Q_K$ contains $(j,k)$, then it will also
    contain $(j,k+1), \ldots, (j,\vrank_i(j)-1)$.  In case $|Q(i)| < K$, we
    will set $Q_K$ to be equal to $Q(i)$. It is easy to see that 
    $\|\vrank_i - \vrank_{i-1} \|_1$ is at most $K$.

\subsection{The Final Algorithm}

The final constant-budget algorithm is the following. Initially, $T_0$
is just the root vertex $0$.  Given a valid tree $T_{i-1}$ with respect
to the admissible virtual rank function $\vrank_{i-1}$, we obtain $T_i$
as follows. We first run the clustering algorithm to get $\rank_i$. Then
we construct the virtual rank function $\vrank_i$ as described in the
previous section, and finally construct a valid tree $T_i$ with respect to
$\vrank_i$.  
\lref[Lemma]{lem:addition} and \lref[Lemma]{lem:modify} imply that we
can construct $T_{i}$ from $T_{i-1}$ by adding one edge and swapping at
most $K$ edges---indeed, we can go from $\vrank_{i-1}$ to $\vrank_{i}$ 
by decrementing (by one) at most $K$ coordinates iteratively, and adding one new 
coordinate for the arriving vertex. 
The algorithm outputs $T_i$ at the end of each round
$i$.

\section{Analysis}
\label{sec:analysis}

The constant number of swaps is enforced by the very definition of the
virtual rank function, so it remains show that for each $i$, the
cost of the tree $T_i$ is close to the cost of the optimal Steiner tree
at the end of round $i$, i.e., $\weight_i(\vrank_i) \approx
\weight_i(\rank_i)$. One approach is to ensure the functions $\vrank_i$
and $\rank_i$ remain close \emph{coordinate-wise close}. We do not know
how to ensure this, but we do achieve closeness in cost.

Let us first give an overview of the proof.
The main problem is that on arrival of vertex $(i+1)$, there may be many
vertices in $[i]$ whose ranks decrease, i.e., $\rank_{i+1}(j) =
\rank_i(j)-1$. Since $\vrank_i$ and $\vrank_{i+1}$ can differ in only a
constant number of locations, the virtual ranks now trail the actual
ranks in many places. And this could potentially happen repeatedly.  Our
first technical result is that such bad events cannot happen in quick
succession. We show that if rank of a vertex $j$ decreases by at least
$2$ between two rounds $i$ and $i'$---i.e., $\rank_{i'}(j) \leq
\rank_i(j) - 2$---then many arrivals must have happened {\em close} to
$j$ after round $i$. We charge the rank decrease of $j$ to one such
arrival after round $i$, and prove that this charging can be done so
that any arrival gets charged only a constant number of times.  More
formally, for every pair $(j,k)$ which denotes that the rank of vertex
$j \in [n]$ has fallen to $k \in [(\rank_n(j)-2)\, \ldots\, (\Init(j) -
1)],$ we {\em charge} a vertex $l$ that arrives after the rank of $j$
drops to $k$; moreover, each vertex $l$ gets charged only $K$ times. The
proof of this charging lemma appears in
\lref[Section]{sec:charging-argument}.

How can we use this charging argument to define the $\vrank_i$ function
for each round $i$? As a thought experiment, suppose we were allowed a
small amount of look-ahead. Then we could proceed as follows~: suppose
$\rank_i(j)$ decreases by 1 in round $i+1$, i.e., $\rank_{i+1}(j) =
\rank_{i}(j)-1$, and we {\em know} that the rank of $j$ will decrease by
at least 1 more in future.  If $\rank_i(j) < \vrank_i(j)$, we add $(j,
\rank_i(j)), \ldots, (j, \vrank_i(j)-1)$ to a queue. 
In round $i+1$, we pick any $K$ pairs from the
queue; if we pick $(j, k)$, we decrease the virtual rank of $j$ to $k$.
Using the above charging argument, we can show that following such a
strategy means the functions $\vrank_n$ and $\rank_n$ differ by at most
two (additively) in each coordinate. And \lref[Lemma]{lem:cost} then
implies that the cost of our tree is within a constant of the optimal
Steiner tree on $[n]$.

Unfortunately, we do not have the luxury of this look-ahead, and so we
instead follow a greedy strategy: in any round, among all the pairs $(j,k)$
in the queue, we pick the ones with highest $k$ (thereby decreasing the
cost of the tree by the maximum possible). By a careful matching-based
argument given in \lref[Section]{sec:endgame}, we show that this
strategy indeed works. We look at an arbitrary round $n$ for rest of the
analysis---our goal is to compare the cost of the tree $T_n$ constructed by our algorithm
at the end of round $n$, and the optimal cost $\Opt([n])$.

\subsection{The Charging Argument}
\label{sec:charging-argument}

In this section, we describe the charging scheme, which charges every
rank decrease of a vertex (except the two most recent rank decreases for each
vertex) to one of the subsequent arrivals, such that each arrival is
charged at most $K = 2\alpha^2$ times.
Formally, we will prove the
following result.

\begin{theorem}
  \label{thm:charge-main}
   Let $L$ be the set of the (``not so
  recent'') rank decreases until this point, and defined as follows:
  \[ L := \bigcup_{j \in [n]} \{ j^{(k)} \mid k \in [\, (\rank_n(j) + 2)\, \ldots
  \, (\Init(j)-1)\,] \} \]
  Then there is a map $F : L \to [n]$ assigning the rank changes to rounds
   such that
  \begin{OneLiners}
    \item[(a)] (constant budget) at most $K = 2\alpha^2$ rank changes from $L$ map to any
      round $i \in [n]$,

    \item[(b)] (feasibility) if $F(j^{(k)}) = i$, then $j$'s rank dropped to $k$ at or before
      round $i$ (i.e., $\rank_{i}(j) \leq k$), and

    \item[(c)] (monotonicity) if $j^{(k)}, j^{(k-1)}$ both lie in $L$, then 
      \stocoption{\\}{}  $F(j^{(k)}) \leq F(j^{(k-1)})$.
  \end{OneLiners}
\end{theorem}
(Note that $j^{(k)}$ is a syntactic object, not $j$ raised to the
power of $k$.)
The proof of this theorem is by constructing a ($b$-)matching in a
suitable bipartite graph.   In
\lref[Section]{sec:witness-lemmas}, we give some technical results which
show that when the rank of a vertex decreases by at least two,
then many new arrivals will happen close to this vertex. 
In \lref[Section]{sec:fract-match}, we describe the bipartite graph, and use the
technical results to prove the existence of a fractional
$b$-matching, and hence an integral $b$-matching, in this graph.

\subsubsection{Disjoint Balls and Witnessing Rank Decreases}
\label{sec:witness-lemmas}

Consider a round $i \in [n]$ and integer $k \geq 0$. Let $A_{ik}$ be the set of vertices $j \in
[n]$ which satisfy one of the following two conditions: either (i)~$j
\in [i]$ and $\rank_i(j) \geq k$, or (ii)~$j \in [n]\setminus[i]$ and
$\Init(j) \geq k$. I.e., $A_{ik}$ has all those nodes that have arrived
by round $i$ and have rank at least $k$ in that round, or will have an
initial rank at least $k$ when they arrive in the future.

For vertex $v \in [i]$, define $\kappa_{i,t}(v)$ to be the cluster in
$\C_i(t)$ containing $v$. We extend this definition to the nodes arriving after round $i$ by defining $\kappa_{i,t}(v)$ for such a node $v \in [n]\setminus[i]$ to be the singleton set $\{v\}$.

\begin{lemma}
  \label{lem:disj}
  For round $i$ and rank $k$, the balls \stocoption{\\}{} $B(\kappa_{i,k}(j),
  \alpha^{k+1})$ for $j \in A_{ik}$ are disjoint.
\end{lemma}

\begin{proof}
  Let $j, l \in A_{ik}$. First assume that $j, l\in [i]$. Then the clusters in
  $\C_{i}(k)$ containing $j$ and $l$ respectively are disjoint (because
  they both have rank at least $k$, by the definition of $A_{ik}$), and
  hence the distance between them is at least $2 \alpha^{k+1}$.

  Now suppose at least one of $j,l$ does not lie in $[i]$, and say $l$
  arrives after $j$. The result follows directly from
  \lref[Claim]{clm:init-rank}.
\end{proof}

The next lemma says that for a vertex $j$ with ``high'' rank $k$ in some
round such that its rank subsequently falls below $k$, this decrease in
rank is witnessed by arrivals close to $j$'s cluster, whose initial
ranks (collectively) are large.
\begin{lemma}
  \label{lem:arrival}
  For round $i$ and rank $k$, suppose $j \in A_{ik}$. Moreover, suppose
  $\rank_n(j) \leq k-2$. Let $X := \{i+1, \ldots, n\} \setminus \{j\}$,
  and let $Y := X \cap B(\kappa_{i,k}(j), \alpha^{k+1})$ be the points
  in $X$ that are within distance $\alpha^{k+1}$ of $j$'s component in
  $\C_i(k)$. Then,
  \begin{gather}
    \sum_{l \in Y} \alpha^{\Init(l)} \geq \alpha^{k-2}. \label{eq:charge}
  \end{gather}
  Furthermore, for any vertex $l \in Y$, $\Init(l) \leq k$.
\end{lemma}

\begin{proof}
  First consider the case when $j \in [i]$, and hence $X = [n] \setminus
  [i]$. Let $C$ represent $j$'s component $\kappa_{i,k}(j)$ in
  $\C_i(k)$. Since $\rank_i(j) \geq k$, $j$ is the leader of $C$, and
  hence it arrived before all other vertices in $C$. For a set $S$ and
  parameters $b,b'$, define the {\em annulus} $B(S,b,b')$ to be $\{ l
  \mid d(l,S) \in (b, b'] \}$---note the half-open interval in the
  definition. In case $b \geq b'$, note that the annulus $B(S,b,b')$ is
  empty.

  Observe that $B(C, 0, 2 \alpha^{k+1}) \cap [i] = \emptyset$, just
  because if there were any other cluster within distance $2\alpha^{k+1}$
  of $C$, we would have merged $C$ with this cluster during phase $k$ of
  $\Run_i$. Let the vertices
  in $Y$---i.e., those which arrive after round $i$ in $B(C,
  \alpha^{k+1})$ be $l_1, l_2, \ldots, l_s$. Let us use $l_0$ to denote
  the vertex $i$. For an index $u \leq s$, let $\Delta_u$ denote the
  cumulative value $\sum_{u'=1}^u \alpha^{\Init(l_{u'})}$ (we define $\Delta_0$
  as 0). Let $l_{[u]}$ denote the set of vertices $\{l_0, l_1, \ldots, l_u \}$. 

  \begin{claim}
    \label{clm:annulus}
    For all $u \in [0\ldots s]$, the annulus \stocoption{\\}{} $B(C, 2\alpha^2 \Delta_u,
    \alpha^{k+1} - 2\alpha^2 \Delta_u)$ does not contain any vertex from
	$l_{[u]}$.
  \end{claim}

  \begin{proof}
    The proof is by induction on $u \in [0\ldots s]$. The base case is
    when $u = 0$, where the claim follows from $B(C, 0, 2 \alpha^{k+1})
    \cap [i] = \emptyset$. Now suppose the claim is true for some $u <
    s$. The next vertex to arrive after $l_u$ in $B(C, 0, \alpha^{k+1})$
    is $l_{u+1}$. By induction hypothesis,  at the beginning of round $l_{u+1}$, the annulus
    $B(C, 2\alpha^2 \Delta_u, \alpha^{k+1} - 2\alpha^2 \Delta_u)$ is
    still empty. For a contradiction, suppose $l_{u+1}$ lies in the
    smaller annulus $B(C, 2\alpha^2 \Delta_{u+1}, \alpha^{k+1} -
    2\alpha^2 \Delta_{u+1})$, then the ball of radius $2 \alpha^2
    (\Delta_{u+1} - \Delta_u) = 2 \alpha^{2+ \Init(l_{u+1})}$ around
    $l_{u+1}$ would be empty in round $l_{u+1}$. But
    this contradicts \lref[Claim]{clm:init-rank}.  This proves the
    claim for $u+1$, and hence for all $u \in [0\ldots s]$.
  \end{proof}

  Observe that the inequality (\ref{eq:charge}) asks us to show that
  $\Delta_s \geq \alpha^{k-2}$. For the sake of contradiction, suppose
  $\Delta_s < \alpha^{k-2}$. In this case \lref[Claim]{clm:annulus} says
  that the annulus \stocoption{\\}{} $B(C, 2 \alpha^k, \alpha^{k+1} - 2 \alpha^k)$ is
  empty at the end of round $l_s$; since there are no further arrivals
  in $B(C, \alpha^{k+1})$, the annulus $B(C, 2 \alpha^k, \alpha^{k+1} - 2
  \alpha^k)$ is also empty after round $n$. Since $\alpha \geq \const$, this
  means that $B(C, 2\alpha^k, 4\alpha^k)$ is empty. If $C'$ denotes the ball
  $B(C, 2 \alpha^k)$, the following two properties hold:
  \begin{OneLiners}
  \item The set $C'$ does not contain any vertex from $[i] \setminus C$.
    This is because $B(C, 2 \alpha^{k+1})$ does not contain any vertex
    from $[i] \setminus  C$, proved in the base case of \lref[Claim]{clm:annulus}.
    Moreover, since $j$ was the leader of the cluster $C$ in round $i$,
    $j$ is the earliest vertex in $C'$ as well.
  \item $B(C', 2 \alpha^k) = C'$, just because the annulus \stocoption{\\}{} $B(C,
    2\alpha^k, 4\alpha^k)$ was empty.
  \end{OneLiners}
  Using the latter property and applying \lref[Lemma]{lem:empty} with
  set $S$ set to $C'$, and round $i$ set to $n$, we infer that $j$'s
  cluster $\kappa_{n,k-1}(j)$ must be contained within $C'$, and $j$ is
  the leader of this cluster $\kappa_{n,k-1}(j).$ But this contradicts
  the assumption that $\rank_n(j) \leq k-2$. 

  To show that the initial ranks of $l_1, \ldots, l_s$ are at most $k$,
  observe that $d(l_j, C) \leq \alpha^{k+1}$ by the definition of
  $Y$. Consequently, in $\Run_{l_u}$, $l_u$ must share a
  cluster with at least one vertex of $C$ in the clustering $\C_{l_u}(k)$.
Since all nodes in $C$ are
  from $[i]$ and arrive before $l_u$, $l_u$ cannot be the leader of its
  component. This completes the proof for the case $j \in [i]$.

  The other case is when $j \not \in [i]$. Since $j \in A_{ik}$, its
  initial rank $\Init(j) \geq k$. This means $B(j, 2\alpha^{k+1})$ does
  not contain any vertex from $[j]$ other than $j$ itself---in other
  words, $\kappa_{j, k}(j) = \{j\}$. Now we can use the same arguments
  as above, just starting from round $j$ (since there are no arrivals in
  $B(\{j\}, 2 \alpha^{k+1})$ during rounds $i$ to $j$).
\end{proof}

We can extend \lref[Lemma]{lem:arrival} to subsets of $A_{ik}$ as follows.
\begin{corollary}
  \label{cor:combo-deal}
  For round $i$ and integer $k$, let $S \subseteq A_{ik}$. Moreover, for
  each $j \in S$, assume $\rank_n(j) \leq k-2$. Let $X = [i+1\ldots n]
  \setminus S$, and let $Y := X \cap (\cup_{j \in S} B(\kappa_{i,k}(j),
  \alpha^{k+1}))$. Then,
  \begin{gather}
    \sum_{l \in Y} \alpha^{\Init(l)} \geq |S| \cdot \alpha^{k-2}. \label{eq:charge2}
  \end{gather}
  Furthermore, for any vertex $l \in Y$, $\Init(l) \leq k$.
\end{corollary}

\begin{proof}
  By \lref[Lemma]{lem:disj}, the balls $B(\kappa_{i,k}(j), \alpha^{k+1})$
  are disjoint, so we can define $Y_j$ as $X \cap B(\kappa_{i,k}(j),
  \alpha^{k+1})$ for each $j \in S$, and apply \lref[Lemma]{lem:arrival}
  to each one of them separately.
\end{proof}




\subsubsection{Constructing the Mapping $F$ via a Matching}
\label{sec:fract-match}

We construct a bipartite graph $H=(L,R = [n],E)$. Here the set $L$ is as
described in the statement of \lref[Theorem]{thm:charge-main}; i.e.,
they are of the form $j^{(k)}$ indicating that the rank of $j$ fell to $k$
at some round in the past, and has subsequently fallen to $k-2$ or lower
by the end of round $n$.  The nodes in $R$ simply represent arrivals
$[n]$.  The edge set $E$ is constructed as follows: we have edge
$(j^{(k)},i)$ if $\rank_i(j) \leq k$. (Note that this condition is same as
the feasibility condition of the map $F$ in
\lref[Theorem]{thm:charge-main}). We say that edge $(j^{(k)}, i)$'s rank is
$k$.


Some
notation: let $\Gamma(v)$ denote the set of neighbors of a node $v$ in
$L \cup R$, and $E(v)$ denote the set of edges incident to $v$. The main
result of this section is the following:

\begin{theorem}
  \label{thm:fract-match}
  There exists an assignment of non-negative values $\{x_e\}_{e \in E}$
  to the edges of $H$ such that
  \begin{OneLiners}
  \item[(a)] for any node $j^{(k)} \in L$, $\sum_{e \in E(j^{(k)})} x_e = 1$, and
  \item[(b)] for any node $i \in R$, $\sum_{e \in E(i)} x_e \leq 2 \alpha^2$.
  \end{OneLiners}
  Moreover, if $\alpha$ is integer, these $x_e$ can be chosen to be in $\{0,1\}$.
\end{theorem}

We first note how \lref[Theorem]{thm:fract-match} implies \lref[Theorem]{thm:charge-main}.
\begin{proofof}{\lref[Theorem]{thm:charge-main}}
If $x_e = 1$ for some edge $e=(j^{(k)}, i)$, we define $F(j^{(k)}) = i$. It is easy to check that the
mapping $F$ satisfies the first two requirements of \lref[Theorem]{thm:charge-main}. It remains
to ensure that $F$ satisfies the monotonicity property. Suppose $F(j^{(k)}) > F(j^{(k-1)})$ for some
pair $j,k$. Then by swapping the values of $F(j^{(k)})$ and $F(j^{(k-1)})$, we 
ensure that $F(j^{(k)}) \leq F(j^{(k-1)})$; moreover, this swap preserves the first two properties of $F$.
We iteratively fix all such violations of the monotonicity property this
way.
\end{proofof}

To prove \lref[Theorem]{thm:fract-match}, we partition the edges of $H$
into subgraphs depending on their rank, and set the $x_e$ values for
each of these subgraphs independently. Specifically, for $k \geq 0$,
define the bipartite graphs $H_k = (L_k, R_k, E_k)$ where $L_k = \{ j
\in [n] \mid j^{(k)} \in L \}$, and $R_k = \{ i \in [n] \mid \Init(i) \leq k
\}$. For an edge $(j^{(k)}, i) \in E$: if
$\Init(i) \leq k$ then we get a corresponding edge $(j, i) \in E_k$, else this edge is simply dropped. We now prove the
following stronger lemma about each $H_k$.

\begin{lemma}
  \label{lem:fract-per-k}
  For each $k$, there exists an assignment of non-negative values
  $\{x_e\}_{e \in E_k}$ to the edges of $H_k$ such that
  \begin{OneLiners}
  \item[(a)] for any node $j \in L_k$, if $E_k(j)$ is the set of edges
    incident to $j$, $\sum_{e \in E_k(j)} x_e = 1$, and
  \item[(b)] for any node $i \in R_k$,
	\begin{eqnarray}
      \label{eq:strong}
      \sum_{e \in E_k(i)} x_e \leq \frac{\alpha^{2+ \Init(i)}}{\alpha^k}.
    \end{eqnarray}
  \end{OneLiners}
\end{lemma}

\begin{proof}
  The proof is constructive. We start with $x_e = 0$ for all $e \in
  E_k$. For each right node $i \in R_k$, define its initial potential
  $\Deficit_k(i) = \frac{\alpha^{2+ \Init(i)}}{\alpha^k}$; this
  potential will measure how much $x_e$ value can be assigned in the
  future to edges in $E_k(i)$. Moreover, recall that $i \in R_k \iff
  \Init(i) \leq k$.

  For a vertex $j \in L_k$, let $\Round_k(j)$ be the first round in
  which rank of $j$ becomes $k$---i.e., $\Round_k(j) = \min \{ i
  \mid \rank_i(j) = k\}$. Order the vertices in $L_k$ in
  \emph{non-increasing} order of their $\Round_k(j)$ values---let this
  ordering be $j_1, \ldots, j_s$. Hence $j_1$ is the last vertex to
  achieve rank $k$, and $j_s$ is the first vertex to do so. The
  algorithm in \lref[Figure]{fig:alg} greedily sets the $x_e$ values of
  the edges incident to the vertices in this order.
  \stocoption{\begin{figure}[ht]
      \centering
      \begin{boxedminipage}{\columnwidth}
        {\bf Algorithm Fractional-Matching($H_k$) :} \medskip\\
        \sp Let $j_1, j_2, \ldots, j_s$ be the vertices in $L_k$ in \\
\sp\sp\sp\sp\sp\sp\sp\sp\sp\sp non-increasing order of $\Round_k(\cdot)$. \\
        \sp For $u = 1, \ldots, s$ \\
        \sp \sp Find values of $x_e$ for edges $e \in E_k(j_u)$ such that \\
        \sp \sp \sp(i) $\sum_{e \in E_k(j_u)} x_e = 1$, and \\
        \sp \sp \sp (ii) if $e \in E_k(j_u)$ has right end-point $i$ in
        $R_k$, \\
\sp\sp\sp\sp\sp\sp\sp\sp\sp\sp then $x_e \leq \Deficit_k(i)$.  \\
        \sp \sp For each $(j_u, i) \in E_k(j_u)$, decrease $\Deficit_k(i)$ by $x_e$.
      \end{boxedminipage}
      \label{fig:alg}
      \caption{The fractional matching algorithm for a fixed $k$}
  \end{figure}}
  {\begin{figure}[ht]
      \centering
      \begin{boxedminipage}{\textwidth}
        {\bf Algorithm Fractional-Matching($H_k$) :} \medskip\\
        \sp \sp Let $j_1, j_2, \ldots, j_s$ be the vertices in $L_k$ in non-increasing order of $\Round(\cdot)$. \\
        \sp \sp For $u = 1, \ldots, s$ \\
        \sp \sp  \sp \sp Find values of $x_e$ for edges $e \in E_k(j_u)$ such that \\
        \sp \sp \sp \sp \sp \sp(i) $\sum_{e \in E(j_u)} x_e = 1$, and \\
        \sp \sp \sp \sp \sp \sp (ii) if $e \in E_k(j_u)$ has right end-point $i$ in $R_k$, then $x_e \leq \Deficit_k(i)$.  \\
        \sp \sp \sp \sp For each $(j_u, i) \in E_k(j_u)$, decrease $\Deficit_k(i)$ by $x_e$.
      \end{boxedminipage}
      \label{fig:alg}
      \caption{The fractional matching algorithm for a fixed value of $k$}
  \end{figure}}

  Observe that if the algorithm terminates successfully, we have an
  assignment of $x_e$ values satisfying properties~(a) and~(b). Indeed,
  property~(a) is guaranteed by property~(i) of the algorithm, and the
  inequality~(\ref{eq:strong}) follows from the fact that the potential
  $\Deficit_k(i)$ captures exactly how much the $x_e$ values can be
  decreased without violating~(\ref{eq:strong}), and these potentials
  never become negative. So it suffices to show that for any vertex $j_u
  \in L_k$, the algorithm can find values $\{x_e\}_{e \in E_k(j_u)}$
  satisfying properties~(i) and~(ii) during iteration~$u$.

  The proof is by induction on $u$. Suppose the algorithm has
  successfully completed the steps for $j_1, \ldots, j_{u-1}$, and we
  are considering $j_u$. Consider $\Round_k(j_u)$, the round in which
  rank of $j_u$ first becomes $k$. For the sake of brevity, let $r^\star :=
  \Round_k(j_u)$. Consider any node $j_p$ with $p \leq u$: such a node
  is either $j_u$ itself, or it has already been processed by the
  algorithm. There are two cases:
  \begin{OneLiners}
  \item Case I: $j_p \leq r^\star$; i.e., $j_p$ arrived at or before the
    round in which $j_u$ attained rank $k$. Since $p \leq u$, our choice
    of the ordering on nodes of $L_k$ ensures that $j_p$ itself attains
    rank $k$ in or after this round $r^\star$. Hence its rank in round $r^\star$
    must be at least $k$, and hence $j_p \in A_{r^\star k}$.~\footnote{Recall
      that $A_{ik}$ was defined immediately after
      \lref[Lemma]{lem:empty}.}
  \item Case II: $j_p > r^\star$; i.e., in the round where $j_u$ attained
    rank $k$, $j_p$ has not arrived at all. However, we know $j_p$
    eventually attains rank $k$, so its initial rank $\Init(j_p)$ is at
    least $k$. Consequently, $j_p \in A_{r^\star k}$ in this case as well.
  \end{OneLiners}
  Look at the set $S = \{j_1, j_2, \ldots, j_u\} \subseteq A_{r^\star
    k}$. By the construction of the graph $H$, the final rank of each
  node in $L_k$ is at most $k-2$. Now \lref[Corollary]{cor:combo-deal}
  implies the existence of the set $Y$ of vertices with $\sum_{l \in Y}
  \alpha^{\Init(l)} \geq u \cdot \alpha^{k-2}$. Moreover, it ensures
  that each vertex $l \in Y$ has $\Init(l) \leq k$, and hence belongs to
  $R_k$. Finally, this set $Y \sse [r^\star+1 \ldots n]$, and so $j_u$ has
  edges to all these nodes in $Y \sse R_k$---this follows from the
  observation that $j_u$ achieved rank $k$ in round $r^\star$, and hence
  edges $(j_u^{(k)}, i)$ were added to $E$ for all $i \geq
  r^\star$. Combining all these facts, we infer that the initial
  potential of nodes in the neighborhood of $j_u$ in the graph $H_k$ is
  at least
  \[ \sum_{l \in Y} \frac{\alpha^{2+\Init(l)}}{\alpha^k} \geq
  \frac{u\cdot \alpha^{k-2} \cdot \alpha^2}{\alpha^k} \geq u. \] Since
  each preceding $j_p$ results in a unit decrease in potential, the
  total decrease in the potential of these nodes in $Y$ in previous
  steps is at most $u-1$. Hence, in iteration $u$, the remaining
  potential of the neighbors of $j_u$ must be at least $1$, which means
  the algorithm can always define the $x_e$ values satisfying
  properties~(i) and~(ii).
\end{proof}

Finally, we use \lref[Lemma]{lem:fract-per-k} to complete the
proof of \lref[Theorem]{thm:fract-match}.
\begin{proofof}{\lref[Theorem]{thm:fract-match}}
  Given the graph $H = (L, R, E)$, each rank-$k$ edge $(j^{(k)}, i) \in E$ either
  gives rise to an edge $(j,i) \in H_k$, or is ignored. Independently
  apply \lref[Lemma]{lem:fract-per-k} to each $H_k$ to get an assignment
  of $x_e$ to each edge in $\cup E_k$, and let each edge in $E$ inherit
  the $x_e$ value of its corresponding edge in $\cup E_k$. (If there is
  no such corresponding edge, set $x_e = 0$.)  Since each node $j^{(k)}$
  corresponds to exactly one node  $j \in L_k$,
  \lref[Lemma]{lem:fract-per-k}(a) implies
  \lref[Theorem]{thm:fract-match}(a).

  Now for the fractional degrees of nodes on the right. For each node $i
  \in R$, note that this node $i \in R_k$ only if $\Init(i) \leq k$. All
  edges incident to $i$ that do not belong to any of these graphs have
  $x_e = 0$. Hence, adding~(\ref{eq:strong}) for all values $k \geq
  \Init(i)$, we get 

  \stocoption {
    \begin{align*}
    \sum_{e \in E(i)} x_e & = \sum_{k \geq \Init(i)} \sum_{e \in E_k(i)} x_e \leq \alpha^2 (1 + 1/\alpha + 1/\alpha^2 + \cdots) \\  
                          & = \frac{\alpha^3}{\alpha - 1} \leq 2 \alpha^2.  \end{align*} }
  {
   \[ \sum_{e \in E(i)} x_e = \sum_{k \geq \Init(i)}
  \sum_{e \in E_k(i)} x_e \leq \alpha^2 (1 + 1/\alpha + 1/\alpha^2 +
  \cdots) = \frac{\alpha^3}{\alpha - 1} \leq 2 \alpha^2. \] }
  The last inequality uses the fact that
  $\alpha \geq 2$.  Finally, the statement about $\{0,1\}$ $x_e$ values follows
  from the integrality of the $b$-matching polytope.
\end{proofof}

\subsection{Bounding the Cost}
\label{sec:endgame}

We now show that for any $n \geq 0$, the cost of a valid tree with respect to
the virtual rank function $\vrank_n$ is within a constant of
$\Opt([n])$. Recall the weight function $\weight_n(\cdot)$ defined in~(\ref{eq:wt}).
\begin{theorem}
\label{thm:greedy}
For any round $n \geq 0$, $\weight_n(\vrank_n) \leq \alpha^2 \cdot \weight(\rank_n).$
\end{theorem}

Before we prove this theorem, we use it to prove
\lref[Theorem]{thm:weaker}
\begin{proofof}{\lref[Theorem]{thm:weaker}}
Let $T_n$ be the valid tree with respect to $\vrank_n$
constructed by our algorithm. The result follows from the following inequalities :
\stocoption{
\begin{align*}
 & \cost(T_n)   \stackrel{\small{\mbox{Lemma}}~\ref{lem:cost}}{\leq} 
\frac{2\alpha^3}{\alpha-1} \cdot
\weight_n(\vrank_n) \\ 
 & \ \ \   \stackrel{\small{\mbox{Theorem}}~\ref{thm:greedy}}{\leq}
\frac{2\alpha^5}{\alpha-1} \cdot  \weight_n(\rank_n) \stackrel{\small{\mbox{Lemma}}~\ref{lem:lbd}}{\leq}
\frac{2\alpha^5}{(\alpha-1)^2} \cdot \Opt([n]).
\end{align*}
}
{
$$ \cost(T_n)  \stackrel{\small{\mbox{Lemma}}~\ref{lem:cost}}{\leq} 
\frac{2\alpha^3}{\alpha-1} \cdot
\weight_n(\vrank_n)  \stackrel{\small{\mbox{Theorem}}~\ref{thm:greedy}}{\leq}
\frac{2\alpha^5}{\alpha-1} \cdot  \weight_n(\rank_n) \stackrel{\small{\mbox{Lemma}}~\ref{lem:lbd}}{\leq}
\frac{2\alpha^5}{(\alpha-1)^2} \cdot \Opt([n]).
$$
}
\end{proofof}

We now complete the proof of \lref[Theorem]{thm:greedy}.
\begin{proofof}{\lref[Theorem]{thm:greedy}}
Consider the last round $i^\star \leq n$ such that, at the end of round $i^\star$, we
had $\vrank_{i^\star}(j) = \rank_{i^\star}(j)$ for all $j \in [i^\star]$---the most recent
round after which we had no more pending rank reductions. (There
exists such an $i^\star$, since this property is satisfied at the end of
rounds $0$ and $1$.) This means that at the end of round $i^\star$, the tree
$T_{i^\star}$ was valid with respect to the rank function
$\rank_{i^\star}$, and not just the virtual rank function $\vrank_{i^\star}$. Moreover,
in every round $i \in [(i^\star + 1), \ldots, n]$, the algorithm must have done $K$
rank reductions.  Indeed, if
$\vrank_{i-1}(j) > \rank_i(j)$ for some
round $i$ and vertex $j \in [i-1]$, then we add the pair $(j,\rank_i(j))$ to the set $Q(i)$ in
the algorithm for defining virtual ranks (\lref[Figure]{fig:recourse}). Hence, if the
algorithm does less than $K$ rank reductions in some round $i > i^\star$, then $|Q(i)| \leq K$,
and so, $\vrank_i(j)=\rank_i(j)$ at the end of round $i$. But this would contradict the definition of
$i^\star$.

What rank reductions was the algorithm doing (or trying to do)? These
are represented by the set
\[ X := \cup_{j \in [n]} \{ j^{(k)} \mid k \in [\rank_n(j) \ldots
(\rank_{\max(i^\star,j)}(j) - 1)] \} ~~. \] Indeed, at each round $i \in
[(i^\star+1)
\ldots n]$, the algorithm does $K$ of the swaps in this set. 
So now consider a bipartite graph $H'$, where the vertices on the left are $X$,
and there are $K$ vertices $i_1, i_2, \ldots, i_K$ on the right for
every round $i \in [(i^\star+1) \ldots n]$. Put an edge between $j^{(k)}$ and
a unique copy of $i$ if the algorithm reduced $j$'s virtual rank to $k$ in round
$i$. This gives us a matching $M_A$; since the algorithm does $K$ swaps
each round, each node on the right is
matched.

Recall \lref[Theorem]{thm:charge-main}, and the definitions of set $L$
and the map $F: L \to [n]$.  Note that $X \not\sse L$ in general, since
$L$ does not contain $j^{(\rank_n(j))}$ and $j^{(\rank_n(j)+1)}$ for each
$j$. However, $X \sse L \cup \{ j^{(\rank_n(j))}, j^{(\rank_n(j)+1)} \mid j
\in [n] \}$. The map $F$ clearly maps $X \cap L$ to $[n]$; we claim $F$
maps $X \cap L$ to $[(i^\star + 1)\ldots n]$. To see this, consider $j^{(k)}
\in X \cap L$: $j$ achieves rank $\rank_{i^\star}(j) - 1$ only in round
$i^\star + 1$ or later, and hence the feasibility property of $F$ says
that $F(j) \geq i^\star + 1$, which proves our claim.
Hence, we can think of $F$ as giving another matching $M_F$ on the bipartite graph $H'$ defined above:
if $F(j^{(k)})= i$, then we add an edge between $j^{(k)}$ and a distinct one of
the $K$ copies of $i$; we need 
at most $K$ copies due to the ``constant budget'' property of $F$. Furthermore,
$M_F$ matches every vertex in $X \cap L$ to some vertex on the right.

Given these two matchings, $M_A$ capturing the algorithm's behavior, $M_F$
encoding the ``suggested schedule'' given by the mapping $F$, it is natural to
consider their symmetric difference $M_A \symdif M_F$, which consists of
paths and cycles and argue about these. It will be convenient to introduce one last piece of
notation. 
For an edge $e=(j^{(k)},i_s) \in M_F \cup M_A$, where $i_s$ is one of the $K$ copies of $i$,
we associate the quantity $\deficit(e) := \alpha^{k+1}-\alpha^k$. Since
we 
think of $(j^{(k)}, i_s)$ as reducing the rank of $j$ from $k+1$ to $k$ in
round $i$, and so $\deficit(e)$
intuitively denotes the reduction in the total cost by this rank
reduction in round $i$. The following easy-to-prove claims make this
intuition formal. 
\begin{claim}
\label{clm:matching}
For the matching $M_A$,
$$\weight_n(\vrank_n) = \weight_n(\vrank_{i^\star}) + \sum_{i=i^\star+1}^n \alpha^{\Init(i)} - \sum_{e \in M_A}
\deficit(e),$$
whereas for the matching $M_F$,
$$ \alpha^2 \cdot \weight_n(\rank_n) \geq  \weight_n(\rank_{i^\star}) + \sum_{i=i^\star+1}^n \alpha^{\Init(i)} - \sum_{e \in M_F} \deficit(e).$$
\end{claim}
\begin{proof}
We argue about $M_A$ first. Consider a round $i \geq i^\star + 1$, and let the neighbors of the $K$ copies of $i$ (in the matching $M_A$) be
$N(i)=\{j_1^{(r_1)}, \ldots, j_K^{(r_K)}\}$ (note that several  of the $j_u$ could correspond to the same vertex). 
If a vertex $j$ happens to be one of $\{j_1, \ldots, j_K\}$, then 
 $\vrank_{i}(j) = \min \{ r_u : j = j_u \}$. Else, $\vrank_i(j)
   = \vrank_{i-1}(j)$.
Also,  $\{ (j,r): r \in \{ \vrank_{i}(j), \ldots, 
\vrank_{i-1}(j)-1\} \}$ is a subset of $N(i)$. 
In the matching $M_A$, each of the vertices $j_u$ will be matched to a unique copy of $i$, say it is $i_u$.
Furthermore, one new arrival (vertex $i$) happens during round $i$.
Therefore,
\stocoption{
\begin{align*}
& \weight_{i-1}(\vrank_{i-1})-\weight_i(\vrank_{i})  \\
& \quad = \sum_{j \in \{j_1, \ldots, j_K \} } \left( \alpha^{\vrank_{i-1}(j)} - \alpha^{\vrank_{i}(j)}\right)   - \alpha^{\Init(i)} \\
& \quad =  \sum_{j \in \{j_1, \ldots, j_K \} } \sum_{u : j=j_u} \left( \alpha^{r_u+1} - \alpha^{r_u}\right)   - \alpha^{\Init(i)} \\
&  \quad =  \sum_{u=1}^K \deficit((j_u^{(r_u)},i_u))   - \alpha^{\Init(i)}
\end{align*}
}
{
\begin{eqnarray*}
\weight_{i-1}(\vrank_{i-1})-\weight_i(\vrank_{i}) & = & 
\sum_{j \in \{j_1, \ldots, j_K \} } \left( \alpha^{\vrank_{i-1}(j)} - \alpha^{\vrank_{i}(j)}\right)   - \alpha^{\Init(i)} \\
& = & \sum_{j \in \{j_1, \ldots, j_K \} } \sum_{u : j=j_u} \left( \alpha^{r_u+1} - \alpha^{r_u}\right)   - \alpha^{\Init(i)} \\
& = & \sum_{u=1}^K \deficit((j_u^{(r_u)},i_u)) - \alpha^{\Init(i)}
\end{eqnarray*}
}
Summing the above for all $i=i^\star + 1, \ldots, n$ gives us the result
for $M_A$. Now consider $M_F$. For any vertex $j$, let $i_0$ denote
$\max(j, i^\star)$. Then matching $M_F$ matches all the vertices in
$\{ j^{(\rank_{i_0}(j)-1)}, \ldots, j^{(\rank_n(j)+2)}\}$ (which is empty if
$\rank_{i_0}(j)-1 < \rank_n(j)+2$).  The sum of the deficit of all
these edges is exactly
$\alpha^{\rank_{i_0}(j)}-\alpha^{\rank_{n}(j)+2}.$ Summing this over all
$j$ gives us the second part of the lemma.
\end{proof}

\begin{lemma}
\label{lem:end}
$\sum_{e \in M_A} \deficit(e) \geq \sum_{e \in M_F} \deficit(e).$
\end{lemma}
\begin{proof}
Recall that the bipartite graph $H'$ has vertices from $X$ on the left and $[(i^\star+1) \ldots n] \times K$ on the right.
Also $M_A$ completely matches the vertices on the right of $H'$.
The  symmetric difference $M_A \symdif M_F$ of the two matchings
consists of paths and cycles. Any cycle means that the total deficit of the edges from $M_A$ and
$M_F$ in it are equal. What about a path? Since
$M_A$ matches every vertex on the right, we know that any path is either
of odd length (ending with $M_A$-edges), or of even length with both the
end-points being on the left. In the former case, the
total deficit of edges from $M_A$ in this path is at least that of edges from $M_F$.
So the remaining case is when the
 path $P$ consists of an equal number of
(alternating) edges from $M_F$ and $M_A$, and the end-points of $P$ lie on the left side of $H'$.

Let the vertices from the left side of the bipartite graph in this even path $P$ be $x_0, x_1, x_2,
\ldots , x_s$ (in order of their appearance in $P$). Each $x_u$ is of the form $j_u^{(k_u)}$, and let us define $\val(x_j) := k_j$. 
Assume w.l.o.g. that $x_0$ is matched
by $M_F$ (and unmatched in $M_A$) and $x_s$ is matched by $M_A$ (and unmatched in $M_F$). Since $x_1,
\ldots, x_{s-1}$ are matched in both $M_A$ and $M_F$, 
\stocoption{\begin{align}
 \nonumber
    & \sum_{e \in P \cap M_A} \deficit(e) - \sum_{e \in P \cap M_F} 
  \deficit(e) \\
  \label{eq:3}
    & \quad = \left(\alpha^{k_s+1}-\alpha^{k_s} \right) -
  \left(\alpha^{k_0+1}-\alpha^{k_0} \right).
\end{align}}{\begin{gather}
  \sum_{e \in P \cap M_A} \deficit(e) - \sum_{e \in P \cap M_F}
  \deficit(e) = \left(\alpha^{k_s+1}-\alpha^{k_s} \right) -
  \left(\alpha^{k_0+1}-\alpha^{k_0} \right). \label{eq:3}
\end{gather}}
Thus, we will be done if we show that $k_s \geq k_0$; i.e., $\val(x_s)
\geq \val(x_0)$. We prove this next. 

For any matching $M$ and node $u$ that is matched in $M$, let $M(u)$ be 
the other endpoint of the matching edge containing $u$.
Let $I$ be the right vertices (arrivals) on the path $P$, say $I = \{ i_1
\leq i_2 \leq \ldots \leq i_s \}$. They may appear in any order on the
path, this ordering is just based on arrivals. For $u \in \{ 0, 1,
\ldots, s\}$, let $I_u = \{i_1,
  i_2, \ldots, i_u\}$; hence $I_0 = \emptyset$ and $I_s = I$. 
Define an auxiliary graph $Q_u$ on the vertex set $\{x_0, x_1, \ldots,
  x_s\}$, by adding an edge between $M_F(i)$ and $M_A(i)$ for each $i \in
I_u$. 
Hence $Q_0$ has no edges, and $Q_s$ is the path $\langle x_0, x_1,
\ldots, x_s \rangle$.  It is
to see  that each component of $Q_u$ is a path
with a ``head'' vertex which is the least
indexed one and which is not yet matched to anyone in $I_u$ by $M_A$, and a
``tail'' which has the highest index and is not matched to anyone in
$I_u$ by $M_F$. As a sanity check, in $Q_0$, each node is both head and
tail of its component. In $Q_s$, there is a single path with $x_0$ as
 the head and $x_s$ as the tail.
We now prove the following lemma by induction on $u$. 
\begin{claim}
  \label{clm:head-tail}
  For each $u \in \{0,1,\ldots, s\}$ and each component of $Q_u$, value of the tail of this
component is at least that of its head. 
\end{claim}

\begin{proof}
The base case is trivially true. Now suppose the claim is true for
 $Q_{u-1}$, and we get $Q_u$ by adding an edge between $M_F(i_u) = x_a$
   and $M_A(i_u)=x_{a+1}$. It must be that $x_a$ is the tail of its
   component with, say, $x_p$ as the head. And $x_{a+1}$ is the head of
   its component, say $x_q$ is the tail.  By the I.H.,
   \begin{gather}
     \val(x_q) \geq \val(x_{a+1}). \label{eq:hdtl}
   \end{gather}
   Moreover, when the algorithm was choosing the rank reductions to
   perform at round $i_u$, both $x_p$ and $x_{a+1}$ were candidates to be
   matched. Since we chose $x_{a+1}$ greedily to have maximum value, we
   have $\val(x_{a+1}) \ge \val(x_p)$. Combining with~(\ref{eq:hdtl}), we
   get $\val(x_q) \ge \val(x_p)$. Since the new component has head $x_p$
   and tail $x_q$, this proves the inductive step and hence \lref[Claim]{clm:head-tail}.
\end{proof}

Applying \lref[Claim]{clm:head-tail} for $u=s$, we get $\val(x_s) \geq
\val(x_0)$, and hence the total deficit of edges in $P \cap M_A$ is at
least that of edges in $P \cap M_F$ by~(\ref{eq:3}). Summing this over
all alternating paths in $M_A \symdif M_f$ completes the proof of
\lref[Lemma]{lem:end}. 
\end{proof}

By the definition of $i^\star$, we know that $\weight_n(\vrank_{i^\star})
= \weight_n(\rank_{i^\star})$. Moreover, we have 
$\sum_{e \in M_A} \deficit(e) \geq \sum_{e \in M_F}
\deficit(e)$ by \lref[Lemma]{lem:end}. Plugging these into \lref[Claim]{clm:matching}, we get that  
\begin{align*}
\weight_n(\vrank_n) &\leq \weight_n(\vrank_{i^\star}) + \sum_{i=i^\star+1}^n \alpha^{\Init(i)} - \sum_{e \in M_A}
\deficit(e) \\
&\leq \weight_n(\rank_{i^\star}) + \sum_{i=i^\star+1}^n \alpha^{\Init(i)} - \sum_{e \in M_F}
\deficit(e) \\
&\leq \alpha^2 \cdot \weight_n(\rank_n).
\end{align*}
This completes the proof of \lref[Theorem]{thm:greedy} (and of
Theorem~\ref{thm:weaker}). 
\end{proofof}

\ifstoc
\else

\section{Just One Swap}
\label{sec:one-swap}

In the previous section, we proved a weaker version
(\lref[Theorem]{thm:weaker}) of our main theorem
(\lref[Theorem]{thm:main1}): using a constant number $K = 2\alpha^2$ of
swaps per arrival, we could maintain a tree $T_n$ with cost at most some
other constant $C = \frac{2\alpha^5}{\alpha-1}$ times the optimum tree
$\Opt([n])$. We now show how to trade off the number of swaps for the
approximation guarantee, and get a constant-factor approximation while
performing at most a single swap per iteration. It is unclear how to
convert a generic algorithm that performs $K$ swaps and maintains a
$C$-approximate tree into one that performs a single swap and maintains
a $f(C,K)$-approximation---this is another place where our dual-based
proof strategy comes handy, since it allows us to perform such a
conversion.

To understand the new ideas, recall the previous algorithm/analysis used
the following conceptual steps:
\begin{enumerate}
\item We show that all but the two most recent rank decreases for each
  vertex can be scheduled so that at most $K$ changes are performed at
  each step. (\lref[Theorem]{thm:charge-main}.)
\item The algorithm takes the set of all the rank changes that have not
  yet been performed (i.e., the set $Q(i) = \{ (j,k) \mid j \in [i-1], k
  \in [\rank_{i}(j) \ldots (\vrank_{i-1}(j)-1) ] \}$ where the virtual
  rank function $\vrank_{i-1}$ lags behind the real rank function
  $\rank_i$, and greedily chooses $K$ of the most beneficial changes to
  perform. (This defines $\vrank_i$, and is described in
  \lref[Section]{sec:def-vrank}.)
\item Finally, we show (in \lref[Theorem]{thm:greedy}) that this greedy
  process ensures the potential function $\weight(\vrank_i) \leq
  \alpha^2\, \weight(\rank_i)$, and hence the cost of our tree is not much
  more than that of the optimal tree.
\end{enumerate}

The main change to get a single-swap algorithm is this: suppose we don't
try to schedule all the not-so-recent rank changes (as in Step~1 above),
but only some subset of the rank changes, such that two successive rank
changes in this set for any vertex differ by approximately $K$. Since we
are then scheduling approximately $1/K$ as many rank decreases as in
\lref[Theorem]{thm:charge-main}, we can get a version of that theorem
with at most one rank change being mapped to each time step. Now we can
change the algorithm (in Step~2 above) to greedily choose the
\emph{single most beneficial rank change} and perform it. The rest of
the argument would follow pretty much unchanged. Since the virtual rank
functions lags the real rank functions on average by an additive $K$,
the corresponding tree is now $\alpha^{O(K)}$-approximate instead of
being $\alpha^{O(1)}$-approximate. Since $K$ is a constant, we prove
\lref[Theorem]{thm:main1}.

In the rest of the section, we first prove an analog of
\lref[Theorem]{thm:charge-main}, give a modified tree-building procedure which works at a
``coarser'' level, describe the modified algorithm to define
the virtual ranks $\vrank_i()$, and finally describe the changes in the
rest of the arguments due to these modifications.

\subsection{A Modified Charging Theorem}
\label{sec:mod-charging-thm}

Our new algorithm will be interested in those values of ranks of a particular vertex
 which are separated  by multiples of $K$. Motivated by this, we define, for a 
vertex $j$, the set $\Z(j,K)$ as $\{ \Init(j)-lK \mid l \in \Z_{\geq 0} \}$. 

Recall the set $L$ from \lref[Theorem]{thm:charge-main}, and define $L'$, the
sparsified version of $L$, as follows:
\stocoption{
  \begin{align*}
      L' := & \bigcup_{j \in [n]} \{ j^{(k)} \mid k \in \Z(j,K) \\
             & \quad \cap [\, (\rank_n(j) + K + 1)\, \ldots \, (\Init(j)-K)\,] \}
  \end{align*}
}
{
\begin{gather}
  L' := \bigcup_{j \in [n]} \{ j^{(k)} \mid k \in \Z(j,K) \cap [\, (\rank_n(j) + K + 1)\, \ldots \, (\Init(j)-K)\,] \}
\end{gather}
}
(Again, note that $j^{(k)}$ is a syntactic object, not $j$ raised to the
power of $k$.)
\begin{theorem}
  \label{thm:charge-main-mod}
  There is a map $F' : L' \to [n]$ assigning the rank changes to rounds
   such that
  \begin{OneLiners}
  \item[(a)] (unit budget) at most one rank change from $L'$ maps to any
      round $i \in [n]$,

    \item[(b)] (feasibility) if $F'(j^{(k)}) = i$, then $j$'s rank dropped to $k$ at or before
      round $i$ (i.e., $\rank_{i}(j) \leq k$), and

    \item[(c)] (monotonicity) if $j^{(k)}, j^{(k-K)}$ both lie in $L'$, then $F'(j^{(k)}) \leq
      F'(j^{(k-K)})$.
  \end{OneLiners}
\end{theorem}

\begin{proof}
  For each element $j^{(k)} \in L'$, consider the set $S(j^{(k)}) := \{j^{(k)},
  j^{(k-1)}, \ldots, j^{(k-K+1)}\}$ of size $K$. By the definition of $L'$,
  these sets for different elements of $L'$ are disjoint; moreover, each
  set $S(j^{(k)})$ is a subset of the set $L$ (as defined in
  \lref[Theorem]{thm:charge-main}).  Now consider the map $F: L \to [n]$
  given by \lref[Theorem]{thm:charge-main}: this can be viewed as a
  bipartite graph between $L$ and $[n]$ where all nodes in $L$ have unit
  degree and nodes in $[n]$ have degree at most $K$. We delete the nodes
  in $L$
  not belonging to $\cup_{j^{(k)} \in L'} S(j^{(k)})$, and contract nodes in
  each $S(j^{(k)})$ for $j^{(k)} \in L'$ into a single ``supernode''. The
  resulting bipartite graph has left-degree exactly $K$, and the right
  degree at most $K$---and by Hall's theorem, has a matching where every
  supernode on the left is matched. This immediately gives us the map
  $F'$: if the edge out of the supernode for $S(j^{(k)})$ goes to $i \in
  [n]$, we set $F'(j^{(k)}) := i$.

  Property~(a) follows from construction. For property~(b), observe that
  the edge from $S(j^{(k)})$ to $i$ in the contracted graph is inherited
  from the fact that $F(j^{(k - c)}) = i$ for some $c \in [0\ldots K-1]$;
  by \lref[Theorem]{thm:charge-main}(b), this means $\rho_i(j) \leq k-c$
  and hence at most $k$. Finally, the monotonicity of $F'$ follows from
  that of $F$.
\end{proof}

\subsection{Modified Tree Building Procedure}
\label{sec:timed-pd-mod}

In this section, we describe the modified tree building procedure. As in
Section~\ref{sec:timed-pd},
we explain the process of maintaining a tree with respect to a rank function.
We say that a  function $\bb: [i] \to \Z_{\geq 0}$ is {\em $K$-admissible} if, for every $j \in [i]$,
$\bb(j) \in \Z(j,K)$ and $\bb(j) \geq \rank_i(j)$. 
 As before, for a subset $S \subseteq [i]$, the \emph{head} of $S$
with respect to $\bb$ is the vertex $j \in S$ with highest $\bb(j)$ value (in case of ties,
choose an arbitrary but fixed tie-breaking rule).

A tree $T=([i], E_T)$ is defined to be {\em $K$-valid with respect to
  $\bb$} if we  we can partition the edge set $E_T$ into sets (which we will call
  {\em levels})  $E_T^1, E_T^2,
\ldots, E_T^r$ such that the following two conditions are satisfied for
each $l \in [1\ldots r]$:
\begin{OneLiners}
\item[(i)] Let $E_T^{\leq l}$ denote $E_T^1 \cup \cdots \cup E_T^l$. For
  any connected component of $E_T^{\leq l}$, let $j$ be the head of this
  component. Then we require $\bb(j) \geq lK$.
\item[(ii)] Each edge in $E_T^l$ has length at most $2 \alpha^{(lK+1)}$.
\end{OneLiners}

Roughly, a level in a $K$-valid tree can be thought of as union of $K$ consecutive
levels of a valid tree (as defined in Section~\ref{sec:timed-pd}). We now prove the
analogues of Lemma~\ref{lem:cost},~\ref{lem:modify},~\ref{lem:addition}. As before
$\weight_i(\bb)$ is defined as $\sum_{j=1}^i \alpha^{\bb(j)}.$

\begin{lemma}
  \label{lem:mod-cost}
  Let $T$ be any tree which is $K$-valid with respect to $\bb$. Then the total cost of
  $T$ is at most
     $2 \frac{\alpha^{2K+1}}{\alpha^K-1} \cdot \weight_i(\bb)$.
\end{lemma}
\begin{proof}
The proof goes exactly as the proof of Lemma~\ref{lem:cost}. The cost of edges in $E_T^l$ can be charged to the heads of the components
of $E_T^{\leq l}$ except for the root vertex. Therefore, a vertex $j$ may get charged for $E_T^l, l=1, \ldots, \left \lfloor \frac{\bb(j)}{K} \right \rfloor +1$. Hence, condition (ii) 
above implies that the total charge to $j$ is at most 
$$ \sum_{l=1}^{\left \lfloor \frac{\bb(j)}{K} \right \rfloor+1} 2 \alpha^{(lK+1)} \leq 2 \frac{\alpha^{2K+1}}{\alpha^K-1} \cdot \alpha^{\bb(j)}.$$ 
This completes the proof of the lemma. 
\end{proof}

Now, we prove the Lipschitz property of $\bb$. 

\begin{lemma}
  \label{lem:mod-modify}
  Let $\bb$ be a $K$-admissible function and $T = ([i], E_T)$ be a $K$-valid
  tree with partition $(E_T^1, \ldots, E_T^r)$. Let $j^\star \in [i]$,
  and suppose $\bb'$ satisfies $\bb'(j) = \bb(j)$ if $j \neq j^\star$,
  and $\bb'({j^\star}) = \bb({j^\star})-K$. Assume that $\bb'$ is
  also admissible (i.e., $\rank_i(j^\star) \leq \bb'({j^\star})$). Then
  there is a valid tree $T'=([i], E_{T'})$ with respect to $\bb'$ such
  that $|E_{T'} \symdif E_T| \leq 2. $
\end{lemma}
\begin{proof}
The proof goes along the same lines as Lemma~\ref{lem:modify}. The only change is when (using the
notation of proof of Lemma~\ref{lem:modify}) $j^\star$ is the head
  of a connected component $C$ in $E_{T'}^{\leq l^\star}$. As before, assume this component
   $C$ now violates condition (i). Hence $\bb'(j^\star) < l^\star K$. 
    Again, we argue that there some vertex $j$, 
   $j \notin C$ and $d(j, C) \leq 2 \alpha^{(l^\star K + 1)}$. Indeed,
   if not, then  \lref[Lemma]{lem:empty} implies
   that there is a vertex $v \in
  C$ such that $\rank_i(v) \geq l^\star K$, and hence $\bb'(v) \geq \rank_i(v) \geq l^\star K$. But then
  $j^\star$ cannot be the head of $C$, a contradiction.
\end{proof}

Finally, we prove that the Lipschitz property holds when a new vertex gets added. 
\begin{lemma}
  \label{lem:mod-addition}
  Suppose $T$ is a $K$-valid tree on $[i]$ with respect to $\bb$. Consider a
  new function $\bb':[i+1]\to \Z_{\geq 0}$ defined thus: $\bb'(j) :=
  \bb(j)$ if $j \leq i$, and $\bb'(i+1) := \Init(i+1)$. Then, there is a
  $K$-valid tree $T'$ with respect to $\bb'$ such that $|T' \symdif T| = 1$.
  Moreover, if $\bb$ was $K$-admissible, then so is $\bb'$.
\end{lemma}

\begin{proof}
The proof follows  along the lines of the proof of Lemma~\ref{lem:addition}. 
Let $u^\star$ denote $\Init(i+1)$. We know that if $j^\star \in [i]$
  is the closest vertex to the new vertex $i+1$, then $2 \alpha^{u^\star
    + 1} < d(j^\star, i) \leq 2 \alpha^{u^\star +2}$ (\lref[Claim]{clm:init-rank}).
    Define $l^\star$ as $\lfloor \frac{u^\star}{K} \rfloor$. 
    For $l \neq
  l^\star+1$, we set $E_{T'}^l := E_{T}^l$, and we define
  $E_{T'}^{l^\star+1} := E_{T}^{l^\star+1} \cup \{ (j^\star, i+1) \}$.
  
It is easy to verify that $T'$ is $K$-valid with respect to $\bb'$. Note that the vertex
  $i+1$ is in a singleton component till level $l^\star$ and is not the head of a component after
  this level. Since $\bb'(i+1)= u^\star \geq K l^\star$, we see that condition (i) is satisfied for $T'$.
  Now condition (ii) needs to be checked for the new edge $(j^\star, i+1)$ only. The length of this edge
  is at most 
  $$2 \alpha^{u^\star +2} \leq 2 \alpha^{l^\star K +K+1} = 2 \alpha^{K(l^\star+1)+1}.$$
  Thus condition (ii) is satisfied as well.

  Note that if $\bb$ was admissible, then using the facts that the actual ranks
  of the vertices can never increase, and that $\bb'(i+1) =
  \rank_{i+1}(i+1)$, we get that $\bb'$ is admissible.
\end{proof}

\subsection{Modified Procedure to define Virtual Ranks Function}
\label{sec:mod-vrank}

We now describe the modified algorithm to maintain the virtual ranks. This will
be similar to the algorithm in \lref[Figure]{fig:recourse}, except that the 
virtual ranks, for a vertex $j$, will take values in $\Z(j,K)$ only.
The modified
algorithm is described in \lref[Figure]{fig:recourse-mod}. It is similar to our earlier
algorithm, except that we improve the virtual rank values in multiples
of $K$ only. It is easy to check that for any vertex $j \in [i]$, 
$\vrank_i(j)$  lies in $\Z(j,K)$. 

\stocoption{\begin{figure*}[ht]}{\begin{figure}[ht]}
  \begin{center}
    \begin{boxedminipage}{6.5in}
      {\bf Virtual Modified-Ranks  :} \medskip\\
        \sp \sp 1. Initially, we just have the root vertex 0. Define $\vrank_0(0)=\infty. $\\
        \sp \sp 2. For $i= 1, 2, \ldots$ \\
        \sp \sp \sp \sp (i) Run the clustering algorithm $\Run_i$ to define the rank function $\rank_i$. \\
        \sp \sp \sp \sp (ii) Set  $\vrank_{i}(i)$ as $\Init(i)$. \\
         \sp \sp \sp \sp (iii) Define $Q(i) = \{ (j,\vrank_{i-1}(j)-K) \mid j \in [i-1], \rank_i(j) \leq  \vrank_{i-1}(j)-K  \}. $ \\
         \sp \sp \sp \sp (iv)  Let $(j^\star,k^\star)$ be the highest pair
         (w.r.t. $\prec$) from $Q(i)$. \\
         \sp \sp \sp \sp (v) Define the first $i-1$ coordinates of $\vrank_{i}$ as follows: \\
          \sp \sp \sp \sp \sp \sp \sp $\vrank_{i}(j) := \left\{ \begin{array}{cc}\vrank_{i-1}(j) &
              \mbox{if }  j \neq j^\star \\  
                \vrank_{i-1}(j^\star)-K &
              \mbox{if } j=j^\star \end{array} \right.$ \\
      \end{boxedminipage}
      \caption{Modified algo.\ for virtual ranks; $K = 2 \alpha^2$.}
      \label{fig:recourse-mod}
      \end{center}
    \stocoption{\end{figure*}}{\end{figure}}

\subsection{Modified Version of Theorem~\ref{thm:greedy}}
\label{sec:mod-greedy-thm}
\newcommand{\round}[2]{{\lceil #1 \rceil}_{#2,K}}
We now prove the analogue of \lref[Theorem]{thm:greedy}. 
\begin{theorem}
  \label{thm:greedy-mod}
  Using the new definition of $\vrank_n$, for any round $n
  \geq 0$, $\weight_n(\vrank_n) \leq \alpha^{2K+1} \cdot
  \weight(\rank_n).$
\end{theorem}

\begin{proof}The proof proceeds along the same lines as that of \lref[Theorem]{thm:greedy}. We point out the main 
modifications to the proof of \lref[Theorem]{thm:greedy}. 
For a vertex $j$ and non-negative integer $k \leq \Init(j)$, define $\round{k}{j}$
as the smallest element of $\Z(j,K)$ which is at least $k$. For a round $i$ and rank vector $\rank_i$, 
define the rounded rank vector $\lceil \rank_i \rceil$ as follows : for each $j \in [i]$, 
$\lceil \rank_i \rceil(j) := \round{\rank_i(j)}{j}.$ Since $\nu_i(j)$ values lie in $\Z(j,K)$, it is easy
to check that for any round $i$, $\vrank_i$ is component-wise at least $ \lceil \rank_i \rceil.$
The round $i^\star$ is defined as the last round $i$ in which 
$\vrank_i = \lceil \rank_i \rceil.$ Again, it is easy to check that we will do one rank update in every round
after $i^\star$.

The set $X$ is now defined as 
\[ X := \cup_{j \in [n]} \{ j^{(k)} \mid k \in \Z(j,k) \cap [\rank_n(j) \ldots
(\rank_{\max(i^\star,j)}(j) - 1)] \} ~~. \] In the bipartite graph $H'$, we need to keep just one copy for each 
round. In the matching $M_{A'}$, we have an edge between $j^{(k)}$ and $i$ if our algorithm set $\vrank_i(j)$ 
to $k$ in round $i$. As before, $M_{A'}$ matches all vertices on the right of $H'$ (which represent rounds
$i^\star, \ldots, n$). We can define the matching $M_{F'}$ using the mapping $F'$ given by 
\lref[Theorem]{thm:charge-main-mod}. One can again check that $F'$ maps $X \cap L'$ to $[(i^\star+1) 
\ldots n]$. We have an edge  $(j^{(k)},i)$ in the matching $M_{F'}$ if $F'(j^{(k)})=i$.

We look at the symmetric difference of the two matchings : $M_{F'} \symdif M_{A'}$. For an edge $e 
= (j^{(k)},i) \in M_{F'} \cup M_{A'}$, define $\deficit(e)$ as $\alpha^{k+K}-\alpha^k$. We can now
show the following analogue of \lref[Claim]{clm:matching} holds. For the matching $M_{A'}$, 
\stocoption{
\begin{align}
\nonumber
\weight_n(\vrank_n) \leq & \weight_n(\vrank_{i^\star}) + \sum_{i=i^\star+1}^n \alpha^{\Init(i)} \\
\label{eq:ma}
                        & \quad \quad - \sum_{e \in M_A}\deficit(e),
\end{align}
whereas for the matching $M_{F'}$,
\begin{align}
\nonumber
 \alpha^{K+1} \cdot \weight_n(\lceil \rank_n \rceil) \geq & \weight_n(\lceil \rank_{i^\star} \rceil) + \sum_{i=i^\star+1}^n \alpha^{\Init(i)} \\
\label{eq:mf}  & \quad \quad - \sum_{e \in M_F} \deficit(e).
\end{align}
}
{
\begin{eqnarray}
\label{eq:ma}
\weight_n(\vrank_n) \leq \weight_n(\vrank_{i^\star}) + \sum_{i=i^\star+1}^n \alpha^{\Init(i)} - \sum_{e \in M_A}
\deficit(e),
\end{eqnarray}
whereas for the matching $M_{F'}$,
\begin{eqnarray}
\label{eq:mf}
 \alpha^{K+1} \cdot \weight_n(\lceil \rank_n \rceil) \geq  \weight_n(\lceil \rank_{i^\star} \rceil) + \sum_{i=i^\star+1}^n \alpha^{\Init(i)} - \sum_{e \in M_F} \deficit(e).
\end{eqnarray}
}
The proof of \lref[Lemma]{lem:end} carries over without any changes (using the modified definition of $\deficit(e)$). So,
combining inequalities~(\ref{eq:ma}) and (\ref{eq:mf}), we get 
$$ \weight_n(\vrank_n) \leq \alpha^{K+1} \cdot \weight_n(\lceil \rank_n \rceil). $$
But the vectors $\rank_n$ and $\lceil \rank_n \rceil$ differ by at most $K$ in each coordinate. Hence, 
$\weight_n(\lceil \rank_n \rceil) \leq \alpha^K \weight_n (\rank_n).$ This proves the theorem. 
\end{proof}

Proceeding as in the proof of \lref[Theorem]{thm:weaker}, we get 
$$ \cost(T_n) \leq \frac{2 \alpha^{4K+2}}{(\alpha-1)(\alpha^K-1)} \cdot \Opt([n]). $$
This proves \lref[Theorem]{thm:main1}. 

In fact, we can prove a stronger version of \lref[Theorem]{thm:main1}.
\begin{theorem}
\label{thm:main1-str}
Given a parameter $\delta$, $0 < \delta \leq 1$, there is an online $2^{O(\frac{1}{\delta})}$-competitive
algorithm for metric Steiner tree 
which performs at most one swap upon each arrival, and at most $\delta$ swaps on each arrival 
in  the amortized sense. 
\end{theorem}
\begin{proof} We give a sketch of the proof. Assume without loss of generality that $\frac{1}{\delta}$ is an 
integer. The main idea is again to strengthen \lref[Theorem]{thm:charge-main}. We were able to get a stronger
version of this theorem, i.e., \lref[Theorem]{thm:charge-main-mod}, by grouping vertices of $L$ into groups of 
$K$. 

Let $K'$ denote $\frac{K}{\delta}$. Define 
\stocoption{
\begin{align*}
L'' :=  & \bigcup_{j \in [n]} \{ j^{(k)} \mid k \in \Z(j,K') \\
        & \quad \cap  [\, (\rank_n(j) + K' + 1)\, \ldots \, (\Init(j)-K')\,] \}.
\end{align*}}
{$$L'' :=  \bigcup_{j \in [n]} \{ j^{(k)} \mid k \in \Z(j,K') \cap 
[\, (\rank_n(j) + K' + 1)\, \ldots \, (\Init(j)-K')\,] \}. $$} Let
$[n]_{\delta}$ denote those elements of $[n]$ which are multiples of
$\frac{1}{\delta}$. We can now generalize
\lref[Theorem]{thm:charge-main-mod} even further to show that there
exists a map $F'' : L'' \to [n]_\delta$ such that
\begin{OneLiners}
\item[(a)] (unit budget) at most one rank change from $L''$ maps to any
  round $i \in [n]_\delta$,
  
\item[(b)] (feasibility) if $F''(j^{(k)}) = i$, then $j$'s rank dropped to $k$ at or before
  round $i$ (i.e., $\rank_{i}(j) \leq k$), 
  
\item[(c)] (monotonicity) if $j^{(k)}, j^{(k-K')}$ both lie in $L'$, then $F''(j^{(k)}) \leq
  F''(j^{(k-K')})$.
\end{OneLiners}
The proof again follows that of \lref[Theorem]{thm:charge-main-mod},
where we now group vertices of $L$ into groups of size $K'$ and those of
$[n]$ into groups of size $\frac{1}{\delta}$. Our online algorithm is
same as that in \lref[Figure]{fig:recourse-mod}, with $K$ replaced by
$K'$. Moreover, we perform the steps of this algorithm only for those
rounds $i$ which are multiples of $\frac{1}{\delta}$ (i.e., in Step 2,
if $i$ is not a multiple of $\frac{1}{\delta}$, then we just perform
steps 2(i) and 2(ii)). The proof now proceeds as in that
of \lref[Theorem]{thm:main1}.
\end{proof}

\section{A Tight Amortized Analysis}
\label{sec:all-swaps}

In this section, we analyze the following
greedy algorithm of Imase and Waxman~\cite{IW91}. Given a parameter $\eps > 0$,
their algorithm, which we call $\mathcal{B}_{1+\eps}$, works as follows. It maintains
a tree connecting all the demands which have arrived so far. Let $T_i$ be the constructed
by the algorithm for vertices in $[i]$.
When
the vertex $i+1$ arrives, it first connects $i+1$ to the closest vertex in $[i]$. Moreover,
whenever there is an edge $e$ in the current tree and a non-tree edge $f$ such that
 $\len(e) > (1+\eps)\,\len(f)$ and $T + f
- e$ is also a (spanning) tree, we swap the edges $e$ and $f$, i.e., we add $f$ and remove $e$ from the
current tree. We get the tree $T_{i+1}$ when this swapping process ends.
It is immediate from the construction that the spanning tree
maintained has weight within a factor $(1+\eps)$ of the best spanning
tree. The goal is to show that for any $n$ and constant $\eps$, the number of swaps made in
the first $n$ steps is $O(n)$. 
Clearly, we cannot hope for a better result, because
there are simple examples showing that the arrival of a single vertex  might cause $\Omega(n)$ swaps.

Recently,
Megow et al.~\cite{MSVW12} showed that a close variant of this algorithm
(which ``froze'' edges when they had a very small length and did not
perform any swaps with them) performed at most $O(n/\eps \log 1/\eps)$
swaps.  In this section, we prove \lref[Theorem]{thm:main2} and show that
the algorithm $\mathcal{B}_{1+\eps}$ (without additional freezing operations)
performs only
$O(n/\eps)$ swaps (and at most $2n$ swaps for $\eps = 1$).

\subsection{An Improved Bound for All-Swaps}
\label{sec:improved-bound}

Let $T^\star = \{e_1, \dots, e_n\}$ be a minimal spanning tree on $[n]$.
 Suppose the greedy edges that we add for vertices
$1, \dots, n$ are respectively $g_1^{\circ}, \dots, g_n^{\circ}$.
An edge in the final tree $T_n$ is obtained by a  sequence of swaps starting from
one of the greedy edges $g_r^{\circ}$, for some unique $r$. Thus, we can define a bijection
between the edges in the final tree $T_n$, denoted by $g_1^f, \ldots, g_n^f$, such that for any $r$,
$g_r^f$ is obtained by a sequence of swaps starting from $g_r^{\circ}$.
In this section, we denote the length of an edge $e$ by $c(e)$. 
Since each swap replaces an edge by another that is a factor $(1+\eps)$
shorter, an upper bound on the total number of swaps performed is
\begin{gather}
  \displaystyle\log_{1+\eps} \frac{c(g_1^{\circ})}{c(g_1^f)} + \dots +
  \log_{1+\eps} \frac{c(g_n^{\circ})}{c(g_n^f)} = \log_{1+\eps}
  \frac{\prod_{i=1}^{n} c(g_i^{\circ})}{\prod_{i=1}^{n} c(g_i^f)} \label{eq:1}
\end{gather}

\begin{theorem}
  \label{thm:main-allswaps}
  The quantity $\prod_{i=1}^{n} c(g_i^{\circ})/ \prod_{i=1}^{n} c(g_i^f)$ is
  bounded by $4^n$. Hence 
  the algorithm $\mathcal{B}_{1+\eps}$ performs at most $n
  (\log_{1+\eps} 4) \in O(n/\eps)$, and $\mathcal{B}_{2}$ performs at
  most $2n$ swaps.
\end{theorem}

This result improves on the result of~\cite{MSVW12} who gave a bound of
$O(\frac{n}{\eps} \log \frac{1}{\eps})$ on the number of swaps for their
freezing-based variant of $\mathcal{B}_{1+\eps}$. 
In \lref[Section]{sec:lbd}, we will show an example for which one needs
at least $1.25n$ swaps. Now, to prove \lref[Theorem]{thm:main-allswaps},
let us give a lower bound on $\prod_{i=1}^{n} c(g_i^f)$.
\begin{lemma}
\label{lem:treemap}
$\prod_{i=1}^{n} c(g_i^f) \geq \prod_{i=1}^{n} c(e_i)$.
\end{lemma}
\begin{proof}By well-known properties of spanning trees (and matroids), there exists a bijection $\psi$ between the edges of  $T^\star$ and $T_n$ such that
  for all edges $e \in T^\star$, the graph $\{T^\star \cup \psi(e)\} \setminus e$ is a
  tree~\cite[Corollary~39.12a]{Schrijver-book}. Since $T^\star$ was chosen to have minimal total cost, $c(e) \leq c(\psi(e))$. Therefore $\prod_{i=1}^{n} c(e_i) \leq \prod_{i=1}^{n} c(\psi(e_i)) = \prod_{i=1}^{n} c(g_i^f)$.
\end{proof}

In light of this claim, proving \lref[Theorem]{thm:main-allswaps}
reduces to showing the following lemma.
\begin{lemma}
\label{lem:potential}
$\prod_{i=1}^{n} c(g_i^{\circ}) \leq 4^n \cdot \prod_{i=1}^{n} c(e_i)$.
\end{lemma}

The proof of this lemma will occupy most of the rest of this section. It
is based on a few useful but simple facts, which we prove next.

\begin{lemma}
\label{lem:goodfunction}
There exists a function $f : \mathbb N \to \mathbb R$ such that
\begin{enumerate}[(i)]
\item $f(1) = 1$, and for all $\ell \in \mathbb N$, $f(\ell) \geq 1$
\item For all $\ell \in \mathbb N$,
\[
\displaystyle\sum_{i=1}^{\ell-1} \frac{f(\ell)}{f(i)\cdot f(\ell-i)} \leq 4.
\]
\end{enumerate}
\end{lemma}

The proof of this lemma is based on an unedifying calculation, and is
deferred to \lref[Section]{sec:good-func}. Note that the constant $4$ in
\lref[Lemma]{lem:goodfunction} is 
the same constant $4$ that appears in \lref[Lemma]{lem:potential}; if one is
satisfied with a worse constant, one could use $f(\ell) = \ell^2$, for
which a bound of $\frac{4}{3}\pi^2$ is easy to prove. 

\begin{lemma}
\label{lem:splittingedge}
Consider a tree $T$ with $\ell$ nodes, and a path $P$ on this tree
consisting of edges $h_1, \dots, h_k$ in order. For any edge $e \in T$,
let $\ell_e$ and $\ell_e'$ denote the number of vertices in the two
trees formed by deleting $e$.  Then there exists some edge $h \in P$ such
that
\begin{gather}
  \dfrac{c(P)}{c(h)} \leq 4\, \dfrac{f(\ell_h)\cdot f(\ell_h')}{f(\ell)}, \label{eq:2}
\end{gather}
where $f(\cdot)$ is the function from \lref[Lemma]{lem:goodfunction}.
\end{lemma}

\begin{proof}
Suppose otherwise; then for all $i$, $$c(h_i) < \dfrac{c(P)}{4} \dfrac{f(\ell)}{f(\ell_{h_i}) f(\ell_{h_i}')}.$$ Hence, 
\stocoption{
\begin{align*}
c(P) & = \sum_{i =1}^k c(h_i)
< \sum_{i =1}^k \dfrac{c(P)}{4} \cdot \dfrac{f(\ell)}{f(\ell_{h_i})\cdot f(\ell_{h_i}')} \\
& \quad \leq \dfrac{c(P)}{4} \sum_{j =1}^{\ell-1} \dfrac{f(\ell)}{f(j) \cdot f(\ell-j)}
\leq c(P),
\end{align*}
}
{\begin{align*}
c(P) = \sum_{i =1}^k c(h_i)
< \sum_{i =1}^k \dfrac{c(P)}{4} \cdot \dfrac{f(\ell)}{f(\ell_{h_i})\cdot f(\ell_{h_i}')}
\leq \dfrac{c(P)}{4} \sum_{j =1}^{\ell-1} \dfrac{f(\ell)}{f(j) \cdot f(\ell-j)}
\leq c(P),
\end{align*}}
which is a contradiction. (The third inequality just contains more
non-negative terms than the second one, and the last inequality used
\lref[Lemma]{lem:goodfunction}.)
\end{proof}

\begin{definition}
  \label{def:delta}
  For $1 \leq i \leq n$, let $\Delta_i$ be the smallest number such that
  there exists a partition of $[n]$ into $i$ parts, such that the
  induced subgraph for each part has diameter at most $\Delta_i$.
\end{definition}

\begin{lemma}
\label{lem:kdiam}
For $1 \leq i \leq n$, suppose that $h_i$ is the $i$-th largest greedy edge, so that $h_1, \dots, h_n$ is a permutation of $g_1^{\circ}, \dots, g_n^{\circ}$ and $h_1 \geq \dots \geq h_n$. Then $c(h_i) \leq \Delta_i$.
\end{lemma}

\begin{proof}
For all $1 \leq i \leq n$, let $x_i$ be the vertex associated with $h_i$'s arrival, and define $x_0 = 0$. Then any edge in the subgraph induced by $\{x_0, x_1, \dots, x_i\}$ has cost at least $h_i$. By definition, there exists a partition of $V$ into $i$ components all of diameter at most $\Delta_i$. But two of the $i+1$ vertices $\{x_0, \dots, x_i\}$ must lie in the same component of the partition, so that their distance is at most $\Delta_i$, implying that $h_i \leq \Delta_i$.
\end{proof}

Moreover, note that $\prod_{i=1}^{n} c(g_i^{\circ}) = \prod_{i=1}^{n}
c(h_i)$, so it suffices to bound the latter.
\begin{proofof}{\lref[Lemma]{lem:potential}}
Consider a permutation $e_1, \dots, e_n$ such that for all $k$, $e_k$
lies on the longest path $P_k$ of the forest of $k$ trees formed by deleting
$e_1, \dots, e_{k-1}$ from $T$, and such that $e_k, P_k$ satisfy the
condition~(\ref{eq:2}) in \lref[Lemma]{lem:splittingedge}. Note that since
$P_k$ is the longest path in the forest, hence the diameter of
every component is at most $P_k$. By Definition~\ref{def:delta} and the
fact that the forest is a partition of $[n]$ into $k$ parts,
we get that $\Delta_k \leq P_k$. Consequently,
\[ \dfrac{\Delta_k}{e_k} \leq \dfrac{P_k}{e_k} \leq 4
\dfrac{f(\ell_e) \cdot f(\ell_e')}{f(\ell)}, \] where $\ell$ is the size
of the component that contains $e_k$, etc. Multiplying over all $k$, the
right side telescopes to $\frac{4^n \cdot f(1)^{n+1}}{f(n+1)} \leq
4^n$. Here we used that $f(n+1) \geq 1$ and $f(1) = 1$. Finally, putting
everything together, we get
\stocoption{
\begin{align*}
\frac{\prod_{i=1}^{n} c(g_i^{\circ})}{\prod_{i=1}^{n} c(g_i^f)}
& \leq \frac{\prod_{i=1}^{n} c(g_i^{\circ})}{\prod_{i=1}^{n} c(\ell_i)}
= \frac{\prod_{i=1}^{n} c(h_i)}{\prod_{i=1}^{n} c(\ell_i)}
\leq \frac{\prod_{i=1}^{n} \Delta_i}{\prod_{i=1}^{n} c(\ell_i)} \\
& \leq \frac{4^n\cdot f(1)^{n+1}}{f(n+1)}
\leq 4^n,
\end{align*}}
{\[
\frac{\prod_{i=1}^{n} c(g_i^{\circ})}{\prod_{i=1}^{n} c(g_i^f)}
\leq \frac{\prod_{i=1}^{n} c(g_i^{\circ})}{\prod_{i=1}^{n} c(\ell_i)}
= \frac{\prod_{i=1}^{n} c(h_i)}{\prod_{i=1}^{n} c(\ell_i)}
\leq \frac{\prod_{i=1}^{n} \Delta_i}{\prod_{i=1}^{n} c(\ell_i)}
\leq \frac{4^n\cdot f(1)^{n+1}}{f(n+1)}
\leq 4^n,
\]}
where the first two inequalities above follow from
\lref[Lemmas]{lem:treemap}, \ref{lem:kdiam}, and the remaining two from
the preceding discussion.
\end{proofof}

\subsubsection{Proof of Lemma~\ref{lem:goodfunction}, and its Tightness}
\label{sec:good-func}

\begin{proofof}{\lref[Lemma]{lem:goodfunction}}
We claim the function
\[
f(\ell) =
\frac{(-1)^{\ell+1}}{2\, \binom{\frac{1}{2}}{\ell}} =
\frac{2^{2\ell-1}\; (2\ell-1)}{\binom{2\ell}{\ell}}
\]
satisfies the desired
properties.

\begin{enumerate}[(i)]
\item Expanding the second formula out gives
\stocoption{
\begin{align*}
& \frac{\left(2^\ell \ell!\right) (2^{\ell-1} \ell!) (2\ell-1)}{(2\ell)!} \\
& =
\frac{(2 \cdot 4 \cdot \dots \cdot 2\ell) (2 \cdot 4 \cdot \dots \cdot (2\ell-2) \cdot \ell) (2\ell-1)}{2\ell!} \\
&= \frac{2 \cdot 4 \cdot \dots \cdot (2\ell-2) \cdot \ell}{1 \cdot 3 \cdot \dots \cdot (2\ell-3)}
\end{align*}
}{
\begin{align*}
\frac{\left(2^\ell \ell!\right) (2^{\ell-1} \ell!) (2\ell-1)}{(2\ell)!} &=
\frac{(2 \cdot 4 \cdot \dots \cdot 2\ell) (2 \cdot 4 \cdot \dots \cdot (2\ell-2) \cdot \ell) (2\ell-1)}{2\ell!} \\
&= \frac{2 \cdot 4 \cdot \dots \cdot (2\ell-2) \cdot \ell}{1 \cdot 3 \cdot \dots \cdot (2\ell-3)}
\end{align*}
}
which is greater than $\ell \geq 1$.

\item We give two proofs of this fact. The first is via generating
  functions.  Consider the formal power series $A(x) =
  \displaystyle\sum_{\ell=1}^\infty \dfrac{x^\ell}{f(\ell)}$.

Then
\begin{equation} \label{eqn1}
A(x)^2 = \sum_{\ell=2}^\infty x^\ell\sum_{i=1}^{\ell-1} \frac{1}{f(i)f(\ell-i)}
\end{equation}

But $\displaystyle A(x) = \sum_{\ell=1}^\infty -2\binom{\frac{1}{2}}{\ell}(-x)^\ell = 2 - 2(1-x)^{\frac{1}{2}}$ by the binomial theorem, so
\begin{equation} \label{eqn2}
A(x)^2 = 4(A(x)-x) = \sum_{\ell=2}^\infty \frac{4x^\ell}{f(\ell)}
\end{equation}

Equating coefficients of \ref{eqn1} and \ref{eqn2}, we see that
\[
\displaystyle\sum_{i=1}^{\ell-1} \frac{f(\ell)}{f(i)\cdot f(\ell-i)} = 4.
\]

\textbf{Proof II:} Here is a purely algebraic proof.
Fix an $\ell$. First note that
\[
\sum_0^n \binom{x}{i} \binom{x}{\ell-i} = \binom{2x}{\ell}
\]
for all real $x$;
this identity is a polynomial in $x$ and it holds for all integral $x \geq \ell$, since both sides count the number of ways to choose a subset of $\ell$ people from a room of $x$ boys and $x$ girls.

It follows that
\stocoption{
\begin{align*}
\sum_1^{\ell-1} \frac{f(\ell)}{f(i)f(\ell-i)}
&= \sum_1^{\ell-1} -2 \frac{\binom{\frac{1}{2}}{i}\binom{\frac{1}{2}}{\ell-i}}{\binom{\frac{1}{2}}{\ell}} \\
&= \frac{-2}{\binom{\frac{1}{2}}{\ell}} \left[ \sum_0^\ell \binom{\frac{1}{2}}{i}\binom{\frac{1}{2}}{\ell-i} - \binom{\frac{1}{2}}{0}\binom{\frac{1}{2}}{\ell} \right. \\
& \quad \quad \quad \quad - \left. \binom{\frac{1}{2}}{\ell}\binom{\frac{1}{2}}{0} \right] \\
&= \frac{-2}{\binom{\frac{1}{2}}{\ell}} \left[ \binom{1}{\ell} - 2\binom{\frac{1}{2}}{\ell} \right] \\
&= \frac{-2}{\binom{\frac{1}{2}}{\ell}} \left[-2\binom{\frac{1}{2}}{\ell} \right] = 4
\end{align*}
}
{
\begin{align*}
\sum_1^{\ell-1} \frac{f(\ell)}{f(i)f(\ell-i)}
&= \sum_1^{\ell-1} -2 \frac{\binom{\frac{1}{2}}{i}\binom{\frac{1}{2}}{\ell-i}}{\binom{\frac{1}{2}}{\ell}} \\
&= \frac{-2}{\binom{\frac{1}{2}}{\ell}} \left[ \sum_0^\ell \binom{\frac{1}{2}}{i}\binom{\frac{1}{2}}{\ell-i} - \binom{\frac{1}{2}}{0}\binom{\frac{1}{2}}{\ell} - \binom{\frac{1}{2}}{\ell}\binom{\frac{1}{2}}{0} \right] \\
&= \frac{-2}{\binom{\frac{1}{2}}{\ell}} \left[ \binom{1}{\ell} - 2\binom{\frac{1}{2}}{\ell} \right] \\
&= \frac{-2}{\binom{\frac{1}{2}}{\ell}} \left[-2\binom{\frac{1}{2}}{\ell} \right] = 4
\end{align*}
}
\end{enumerate}
\end{proofof}

We can also show tightness of our technique: there is no function
which can be used to get a constant better than $4$.
\begin{lemma}
  \label{lem:badfunctions}
  There does not exist a function $g : \mathbb N \to \mathbb R$ and a
  constant $0 < C < 4$ such that $g(1) = 1$ and for all $\ell \in \mathbb N$,
  \begin{enumerate}[(i)]
    \item
    $g(\ell) \geq 1$
    \item
    $\displaystyle\sum_{i=1}^{\ell-1} \frac{g(\ell)}{g(i)g(\ell-i)} \leq C$
  \end{enumerate}
\end{lemma}

\begin{proof}
Suppose there existed such a function. Let $A(x) = \displaystyle\sum_1^\infty \frac{x^i}{g(i)}$. Choose $a \in (\frac{C}{4}, 1)$. By condition (i), $\displaystyle A(a) \leq \sum_1^\infty a^i = \frac{a}{1-a}$ is a positive real number. By condition (ii),
\begin{align*}
A(a)^2 &= \left(\sum_1^\infty \frac{a^i}{g(i)} \right)^2 \\
&= \sum_{\ell=2}^\infty a^\ell \sum_{i=1}^{\ell-1} \frac{1}{g(i)g(\ell-i)} \\
&\leq \sum_{\ell=2}^\infty C\frac{a^\ell}{g(\ell)}
\end{align*}
Therefore $A(a)^2 + C\frac{a}{g(1)} \leq CA(a)$, or $A(a)^2 - CA(a) + Ca \leq 0$. But the quadratic $y^2 - Cy + Ca$ has discriminant $C^2 - 4Ca < 0$ and hence no real solutions in $y$. This is a contradiction, therefore no such function $g$ exists.
\end{proof}

\subsection{A Lower Bound on the Potential Function, and on All-Swaps}
\label{sec:lbd}

In this section, we show that algorithm $\mathcal{B}_{2}$ performs
asymptotically more than $n$ swaps. Previously known examples only
showed that $\mathcal{B}_2$ might need to perform at least $n-1$ swaps;
the following example shows that the correct (worst-case) number lies
between $1.25n$ and $2n$.

\begin{figure}[!htbp]
\begin{center}
\includegraphics[scale=1.25]{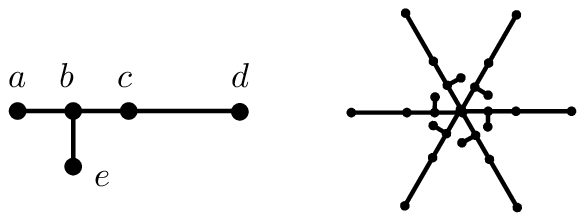}
\end{center}
\end{figure}

Consider the tree on the left where the edges have length $1$, except
$cd$ has length $2$. We take $k = m/4$ copies of it and identify the
vertex $a$ in all $k$ copies to get the tree $T$ on the right. (The
figure shows an example $k=6$.)  There are $4k = m$ edges and $m+1$
nodes in this tree, and the final metric will be the metric closure of
this tree.

Suppose we give the vertices in the following order: we give all the
copies of $d$, then copies of $e$, then of $c$, then the vertex $a$, and
finally the copies of $b$.  Each copy of $d$ (aside from the first) adds
a greedy edge of length $8$.  Each copy of $e$ adds a greedy edge of
length $4$.  The vertex $a$ and each copy of $c$ adds a greedy edge of
length $2$. Finally, each copy of $B$ adds a greedy edge of length $1$.
This means $\prod_i \text{greedy}_i = 8^{k-1} 4^k 2^{k+1} 1^k = 2^{6k -
  2}$. Moreover, after our algorithm finishes, the final tree is just the
tree $T$, the product of whose edge lengths is $2^k 1^{3k} = 2^k$. This
gives us a ratio of $2^{5k - 2} = 2^{1.25\, m - 2}$. Hence we get that
$\prod_i c(g_i) \geq 2^{1.25\, m - 2} \prod_i c(\ell_i)$ for this
instance. Moreover, it is easy to check that the number of swaps
performed by our algorithm is also $1.25 m - O(1)$, which
proves the claim.


\fi

\section{Conclusions}

This paper considers  maintaining an $O(1)$-competitive Steiner
tree in an online environment. In this model, when a new vertex arrives
the distances to previous vertices is revealed, and must form a metric
space. The algorithm is allowed to add an edge connecting this new
vertex to previous vertices, and also to add/delete a constant number of
existing edges. It was previously known that a natural greedy algorithm
makes a total of $O(n)$ additions/deletions and maintains a
constant-competitive tree, which implies that the number of changes per
arrival is constant \emph{on average}. In this paper we give an
algorithm that makes a single change per arrival \emph{in the worst
  case}. Our idea is to use a new constant-amortized-swaps algorithm,
which is then de-amortized by carefully delaying some of the swaps, and
showing that these delays do not result in a significant blowup in cost.
We also give a tight bound and a simpler proof of the natural greedy
constant-average-swaps algorithm.

Several problems remain open: can we show that a primal-only greedy-like
algorithm swap upon each arrival suffices to give
$O(1)$-competitiveness?  (See~\cite{Vers12} for a related conjecture.)
We have not optimized the constants in our result, aiming for simplicity
of exposition, but it would be useful to get a smaller constant factor
that would put it in the realm of practicality. Moreover, can we extend
our algorithm to the case where vertices are allowed to arrive and
depart---the ``fully-dynamic'' case---and get even a constant amortized
bound? Finally, for which other problems can we improve results by
allowing a small number of changes in hindsight? And in what situations
can we use similar de-amortization techniques?

\subsection*{Acknowledgments}

We would like to thank Chaoxu Tong for pointing out an error in the
previous version of Lemma~\ref{lem:modify}.

\ifstoc
\else
 \bibliographystyle{alpha}
 {\small \bibliography{bibonline}}
\fi


\end{document}